\documentclass[12pt, reqno]{amsart}
\usepackage{CJK,CJKnumb}
\usepackage{xcolor}
\usepackage [latin1]{inputenc}
\usepackage{graphicx}
\usepackage{amssymb}
\usepackage{color}
\usepackage{epstopdf}
\usepackage{amsthm,amsmath,amssymb}
\usepackage{mathbbol}
\usepackage{mathrsfs}
\usepackage{cite}
\usepackage{subeqnarray}
\usepackage{cases}
\usepackage{titletoc}

\allowdisplaybreaks

\makeatletter

\newcommand{\Rmnum}[1]{\expandafter\@slowromancap\romannumeral #1@}
\makeatother

\newtheorem{proposition}{Proposition}[section]

\newtheorem{lemma}{Lemma}[section]

\newtheorem{thm}{Theorem}[section]
\newtheorem{definition}{Definition}[section]
\newtheorem{ex}{Example}

\begin{document}
\begin{CJK*}{GBK}{song}

\begin{center}
{\large \sc \bf On the $c$-$k$ constrained KP and BKP hierarchies: the Fermionic pictures, solutions and additional symmetries }

\vskip 20pt
{Kelei Tian\dag, Song Li\dag, Ge Yi\dag, Ying Xu\dag$^*$ and Jipeng Cheng\ddag}	
{\large }

\vskip 20pt

{\it \dag School of Mathematics, Hefei University of Technology, Hefei, Anhui 230601, China \\
\ddag School of Mathematics, China University of Mining and Technology, Xuzhou, Jiangsu 221116, China
 }

\bigskip

$^*$ Corresponding author:  {\tt xuying@hfut.edu.cn }

\bigskip


\end{center}

\bigskip
\bigskip
\textbf{Abstract:} In this paper, we study two generalized constrained integrable hierarchies, which are called the $c$-$k$ constrained KP and BKP hierarchies. The Fermionic picture of the $c$-$k$ constrained KP hierarchy is given. We give some solutions for the $c$-$k$ constrained KP hierarchy by using the free Fermion operators and define its additional symmetries. Its additional flows form a subalgebra of the Virasoro algebra. Furthermore, the additional flows acting on eigenfunctions $q_{i}(t)$ and adjoint eigenfunctions $r_{i}(t)$ of the $c$-$k$ constrained KP hierarchy are presented. Next, we define the $c$-$k$ constrained BKP hierarchy and obtain its bilinear identity and solutions. The algebra formed by the additional symmetric flow of the $c$-$k$ constrained BKP hierarchy that we defined is still a subalgebra of the Virasoro algebra and it is a subalgebra of the algebra formed by the additional flows of the $c$-$k$ constrained KP hierarchy.

\bigskip

\bigskip

\setcounter{tocdepth}{1}

\tableofcontents

\section{Introduction}

Integrable systems constitute an outstanding branch of mathematical physics \cite{HitchinIs}. They have significant applications in conformal maps, topological field theory, matrix model theory and twistor theory \cite{Mineev-WeinsteinIs,KricheverTd,AlexandrovSv,MorozovIa}.
In particular, the Witten conjecture proved by Kontsevich reveals the deep relationship between the geometry of the curve modular space and the integrable system, and it is also the prototype of the the Virasoro conjecture in Gromov-Witten invariants theory \cite{GetzlerTv1999}.
Hirota introduced the Hirota derivative, which can help to write the soliton equation into Hirota bilinear equation and obtain the soliton solution. Sato developed tau function theory under Hirota bilinear method \cite{HirotaNp1}, and he discovered that the solution space of the KP (Kadomtsev-Petviashvili) hierarchy is isomorphic to an infinite dimensional Grassmannian, specifically, the Hirota form of the KP equations can be equivalent to Pl$\ddot{u}$cker relation for the infinite dimensional Grassmannian \cite{SatoSe81}. Usually the theory on the KP and its related hierarchies is called Sato theory, which plays a central role in integrable systems.
In the algebraic aspect of Sato theory, the highest weight representations of the infinite dimensional Lie algebras are very important, which highly relies on the free Fermions. The Lie algebra $A_{\infty}$ operates on the space of the tau functions for the KP hierarchy, that is to say, a tau function can be viewed as a point on the orbit of the Lie group $GL(\infty)$ corresponding to the Lie algebra $gl(\infty)$\cite{KacBl,Kac-vdLeur2023,DateTg4}.
From the viewpoint of the infinite dimensional Lie algebras, the Drinfeld-Sokolov hierarchy, a sub-hierarchies of the KP hierarchy, is widely used in Gromov-Witten theory\cite{WuTf2017,LiuRZBD2015}.
The BKP (Kadomtsev-Petviashvili of $B$ type) and CKP (Kadomtsev-Petviashvili of $C$ type) hierarchies are also two important sub-hierarchies of the KP hierarchy, corresponding to the Lie algebra $B_{\infty}$ and $C_{\infty}$, respectively\cite{JimboS,DateTg4,KacTg,KacPt,Kac-vdLeur2023}. You showed that the DKP (Kadomtsev-Petviashvili of $D$ type ) and BKP hierarchies are essentially the same \cite{YouDa} in a special bosonization.
Since there is no Fermionic representations for the exceptional Lie algebra, it is very difficult to construct the integrable systems by using Kyoto School's methods. Therefore, the Kac school proposed the Kac-Wakimoto construction \cite{KacID,KacBl,KacPts} by using the Casimir operator, which gives an unified way to construct integrable systems associated with the any infinite affine Lie algebras, besides the exceptional ones. Unfortunately, the integrable systems corresponding to Lie algebras of $E$ type have only been partially studied so far, while the study of integrable systems corresponding to Lie algebras of $F$ and $G$ type have made little progress \cite{KacCo,KacEh,MilanovGt}.

The KP hierarchy is one of the most important topics in the field of integrable systems \cite{OrlovA, DKOaSato, DKA, DateN, DKS}. The KP hierarchy is given by
\begin{eqnarray}\label{KPLax}
\frac{\partial L}{\partial t_{n}}=[B_{n},L],\;\;B_{n}=(L^{n})_{+},\;\;n=1,2,3,\cdots,
\end{eqnarray}
where the pseudo-differential operator
\begin{eqnarray}
L=\partial+u_1\partial^{-1}+u_2\partial^{-2}+\cdots,\nonumber
\end{eqnarray}
and the symbols $( A)_{+}$ and $( A)_{-}$  denote $\sum\nolimits_{i=0}^{m} a_{i}\partial^{i}$ and $\sum\nolimits_{i=-\infty}^{-1} a_{i}\partial^{i}$ respectively for arbitrary pseudo-differential operator $A=\sum\nolimits_{i=-\infty}^{m} a_{i}\partial^{i}$. The existence of the tau function of the KP hierarchy is given in \cite{DKS,DateTg4,DateTg1}, which is one of the most important integrability of the KP hierarchy. Chau, Shaw and Yen conveniently use the gauge transformations, which are of both differential and integral types, to give new solutions for the KP hierarchy and naturally to yield new wave functions \cite{ChauSt}. Symmetry, such as squared eigenfunction symmetry and additional symmetry, plays a vital role in the study of integrable systems and it is closely related to conservation laws and Hamiltonian structure \cite{OlverAo}. The squared eigenfunction symmetry, which is important for the study of constrained integrable systems, is a kind of symmetry generated by eigenfunctions and adjoint eigenfunctions \cite{OevelD,AratynMo}. Actually, there is a close connection between these two symmetries \cite{AratynMo}. The additional symmetry was discovered by two different approaches. In the first one, the explicit form of the additional symmetry of the KP hierarchy was given by introducing the Orlov-Schulman operator \cite{OrlovA}, based on the works of Chen, Lee, Lin, Focas and Fuchssteiner \cite{FokasTh,ChenOa}. In the second approach, the additional symmetry was found by the Sato's B$\ddot{a}$cklund transformations of the vertex operator acting on the tau function \cite{DKS,DateTg4,DKOaSato}. The additional symmetry defined in these two ways is coincided by the Adler-Shiota-van Moerbeke formula \cite{AdlerAl,AdlerFt}.  The additional symmetry have profound implications in the development of the string equations, matrix models and Virasoro constraints \cite{DKA,DKOaSato,MorozovIa}.

The KP hierarchy has an important reduction, that is, the $n$-th GD (Gelfand-Dickey) hierarchy \cite{DKS ,AdlerOc ,AdlerOa}. If $n=2$, i.e., pseudo-differential operator $L$ satisfies
\begin{eqnarray}
(L^{2})_{-}=0,\nonumber
\end{eqnarray}
the equation \eqref{KPLax} is called the KdV (Korteweg-de Vries) hierarchy. The first non-trivial equation in the KP hierarchy is the KdV equation, which is the shallow water wave equation proposed by Korteweg and de Vries \cite{KortewegOt}. If $n=3$, i.e., pseudo-differential operator $L$ satisfies
\begin{eqnarray}
(L^{3})_{-}=0,\nonumber
\end{eqnarray}
the equation \eqref{KPLax} is called the Boussinesq hierarchy. Additionally, the KP hierarchy also has an important sub-hierarchy, the constrained KP hierarchy \cite{KonopelchenkoD, ChengT}, given by
\begin{eqnarray}
L^{k}=B_{k}+\sum_{i=1}^{N} q_{i} \partial^{-1}\circ r_{i}, \quad N \geqslant 1, k\in \mathbb{Z}_{+}, \nonumber
\end{eqnarray}
where $\partial^{n}\circ f(x)$ is defined by
\begin{eqnarray}
\partial^{n}\circ f(x)=\sum_{i=0}^{\infty}\binom{n}{i}\frac{\partial ^{i} f(x)}{\partial x^{i}}\partial^{n-i},\;\;n\in \mathbb{Z},\nonumber
\end{eqnarray}
and
\begin{subequations}
\begin{alignat}{2}
q_{i,t_{n}}=B_{n} q_{i},  \quad r_{i,t_{n}}=-B_{n}^{*} r_{i}.\nonumber
\end{alignat}
\end{subequations}
Note that for an arbitrary pseudo-differential operator $P=\sum\nolimits_{i=-\infty}^{m} p_{i}\partial^{i}$, the formal adjoint operator $P^{*}$ is $\sum\nolimits_{i=-\infty}^{m} (-1)^{i}\partial^{i} \circ p_{i}$ and $(AB)^{*}=B^{*}A^{*}$ for the pseudo-differential operators A and B. There are several important results concerning the constrained KP hierarchy \cite{ChengC, DKO, OevelD, AratynC, LiuCi}. In particular, Dickey gave an explanation for the constrained hierarchy by linking it with the symmetries of the KP hierarchy. Dickey pointed out that the existence of ordinary symmetries allow one to reduce the KP hierarchy to the KdV hierarchy, while the existence of additional symmetries allow one to reduce the KP hierarchy to the constrained KP hierarchy \cite{DKO}. Cheng and Zhang gave bilinear identities for the $k$-constrained KP hierarchy and obtained their rational and soliton solutions which are closely related to the tau functions of the KP hierarchy by using the free Fermion operators \cite{ZhangSf}. The Wronskian for the $k$-constrained KP hierarchy is given in \cite{LorisOs, OevelWs}. The $c$-$k$ constrained KP hierarchy \cite{LorisKs} was introduced by
\begin{eqnarray}
L^{k}=B_{k}+\sum_{i=1}^{N} q_{i} \partial^{-1}\circ r_{i}-cL^{-1}, \nonumber
\end{eqnarray}
and an entirely different bidirectional Wronskian is given for the standard case \cite{LorisOs}. The $c$-$k$ constrained KP hierarchy is particularly well suited to treat systems with non-zero boundary conditions. The first aim of this paper is to continue the study of the $c$-$k$ constrained KP hierarchy, which includes its some solutions, Fermionic picture and additional symmetry.

The BKP hierarchy \cite{DateN,DateTg4} is also an important sub-hierarchies of the KP hierarchy.
The BKP hierarchy is a reduction of the KP hierarchy by the constraint
\begin{eqnarray}\label{BKPL}
L^*=-\partial L \partial^{-1},
\end{eqnarray}
and the associated Lax equation of the BKP hierarchy is
\begin{eqnarray}\label{BKPLax}
\frac{\partial L}{\partial t_{2 n+1}}=\left[B_{2 n+1}, L\right], \quad n=0,1,2, \cdots,
\end{eqnarray}
where $B_{2n+1}=\left(L^{2n+1}\right)_{+}$. Similar to the KP hierarchy, some important integrable properties related to the BKP hierarchy have been studied, such as the bilinear identities, tau theory, additional symmetry \cite{KacPt,TuOtB,LeurTasvp,LeurTn}. The constrained BKP hierarchy was given in \cite{ChengC,LorisSr}, and Shen, Lee and Tu gave its solutions in \cite{ShenOa}. The additional symmetry of the constrained BKP hierarchy has been studied in \cite{TianAs}. Following the $c$-$k$ constrained KP hierarchy defined by Loris and Willox  in \cite{LorisKs}, the second aim of this paper is to give the definition of the $c$-$k$ constrained BKP hierarchy and to study its some solutions, Fermionic picture and additional symmetry.

The paper is organized as follows. In Section 2, we will recall some basic definitions and define the $c$-$k$ constrained BKP hierarchy.  In Section 3, The Fermionic picture of the $c$-$k$ constrained KP and BKP hierarchies are given. In Section 4, we give some solutions of the $c$-$k$ constrained KP and BKP hierarchies by using the free Fermion operators. In Section 5, we define additional symmetries of the $c$-$k$ constrained KP hierarchies. In Section 6, conclusions are given.

Due to their extreme importance in mathematics and physics,  Virasoro algebra have been widely studied in the mathematical and physical literature. It is well known that the Virasoro algebra plays an important role in many areas of theoretical physics and mathematics, which occures in the investigation of conformal field theory and has a $\mathbb{C}$-basis $\{L_n,c |n\in \mathbb{Z}\}$ with the nontrivial relations $[L_n, L_m]=(m-n)L_{n+m}$. It can be regarded as the complexification of the Lie algebra of polynomial vector fields on a circle, and also as the Lie algebra of derivations of the ring $\mathbb{C}[z,z^{-1}]$. The centerless Virasoro algebra $\mathbb{C}$-basis $\{L_n|n\in \mathbb{Z}\}$ admits many kinds of extensions, one of these is the Schr\"{o}dinger-Virasoro type Lie algebras firstly introduced in the context of non-equilibrium statistical physics during the process of investigating the free Schr\"{o}dinger equations.

\section{Preliminaries}

In the section, we first introduce two Clifford algebras $\mathcal{A}$ and $\mathcal{A}_B$ based on the charged free Fermions $\{\psi_{i},\psi_{i}^{*}|i \in \mathbb{Z}\}$ and the neutral free Fermions $\{\phi_{i}|i \in \mathbb{Z}\}$\cite{DateN}. Some results for these two free Fermions are given. Then we give a brief introduction of the $c$-$k$ constrained KP and BKP hierarchies\cite{DateN}.

\subsection{The Clifford algebra}
\addcontentsline{toc}{section}{2.1 The Clifford algebra}
In 1878, the Clifford algebras were introduced as the generalization of the quaternion algebras. Their deeper significance became clear after Dirac's discovery of the spin representation in the study of Klein-Gordon equation. The foundation of the theory of spinors lies in the theory of Clifford algebras and their representations. In recent years, with the increasing interest of physicists, the application of Clifford in the integrable systems have come to the foreground.  We take up the constructions of Clifford algebras $\mathcal{A},\mathcal{A}_{B}$ and their representations \cite{DateN}.

\begin{definition}[Clifford algebra $\mathcal{A}$ and $\mathcal{A}_B$]\label{def:Clifford}
The Clifford algebras $\mathcal{A}$ and $\mathcal{A}_{B}$ can be defined as follows, in which $[X, Y]_{+} \triangleq XY+YX$.
\begin{itemize}
\item[(1).]  The Clifford algebra $\mathcal{A}$ is an associative algebra over $\mathbb{C}$ with generators $\{\psi_{i}, \psi_{i}^{*}|i\in \mathbb{Z}\}$ and $\mathbb{1}$ satisfying the relations
\begin{eqnarray}\label{com-r}
\left[\psi_{i}, \psi_{j}^{*}\right]_{+}&=&\delta_{i j}\mathbb{1},\\
\left[\psi_{i}, \psi_{j}\right]_{+}&=&0,\\
\left[\psi_{i}^{*}, \psi_{j}^{*}\right]_{+}&=&0.
\end{eqnarray}
\item[(2).] The Clifford algebra $\mathcal{A}_{B}$ is an associative algebra generated by the generators $\{\phi_i|i\in \mathbb{Z}\}$ and $\mathbb{1} $ satisfying
\begin{eqnarray}\label{neufer}
\left[\phi_i, \phi_j\right]_{+}=(-1)^i \delta_{i,-j}\mathbb{1}.
\end{eqnarray}

\end{itemize}
\end{definition}

In fact, the Clifford algebra $\mathcal{A}_{B}$ is a subalgebra of $\mathcal{A} $ by taking $\phi_{n}=\frac{\psi_{n}+(-1)^{n}\psi_{-n}^{*}}{\sqrt{2}}.$ Classically, the Clifford algebras can  approach to the Lie algebra and its corresponding group by the normal product.

Let Lie algebra
\begin{eqnarray}
gl(\infty)&=&\left\{ \sum_{m,n \in \mathbb{Z}} a_{mn} : \psi_{m} \psi_{n}^{*} : |there\;exists\; N>0 \right.\nonumber\\
&\;&\left.\; such\;that\;a_{mn}=0 \;for\;all\; m,n \;with\; |m-n|>N \right\}\oplus \mathbb{C}\mathbb{1}
\end{eqnarray}
and Lie group
\begin{eqnarray}
G(\infty)=\left \{ e^{X_{1}}e^{X_{2}}\cdots e^{X_{k}}|X_{i} \in gl(\infty)\right \},
\end{eqnarray}
where normal product $:\psi_{m}\psi_{n}^{*}:=\psi_{m}\psi_{n}^{*}-\langle \psi_{m}\psi_{n}^{*} \rangle.$

The special subalgebra $b(\infty)$ of Lie algebra $gl(\infty)$ can be defined as
\begin{eqnarray}
b(\infty)&=&\left\{\sum_{i, j \in \mathbb{Z}} a_{i j}: \phi_i \phi_j: \mid a_{i j} \in \tilde{b}(\infty)\right\} \oplus \mathbb{C}\mathbb{1}, \nonumber\\
\tilde{b}(\infty)&=&\left\{\sum_{i, j \in \mathbb{Z}} a_{i j} E_{i j}|a_{i j}=(-1)^{i+j+1} a_{-j,-i}, \;a_{i j}=0\; for\; |i-j|\gg 0\right\},\nonumber
\end{eqnarray}
and its corresponding group is
\begin{eqnarray}
B(\infty)=\left\{e^{X_1} e^{X_2} \cdots e^{X_l} \mid X_1, X_2, \cdots, X_l \in b(\infty)\right\}.\nonumber
\end{eqnarray}

The following lemma is a well-known fact \cite{AlexandrovZFf,DateN}.
\begin{lemma}
\begin{itemize}
\item[\rm(1).] If $g \in G(\infty)$, then there exists a matrix $A= (A_{mn})_{m,n\in \mathbb{Z}}$ such that
\begin{subequations}\label{KPtr}
\begin{alignat}{2}
&g^{-1}\psi_{n}g=\sum_{m \in \mathbb{Z}} A_{mn} \psi_{m},\\
&g^{-1}\psi_{n}^{*}g=\sum_{m \in \mathbb{Z}} (A^{-1})_{nm} \psi_{m}^{*}.
\end{alignat}
\end{subequations}

\item[\rm(2).] If $g \in B(\infty),\phi_n^*=(-1)^n \phi_{-n}$, then there exists matrix $B= (B_{mn})_{m,n\in \mathbb{Z}}$ and $B^{'}= (B^{'}_{mn})_{m,n\in \mathbb{Z}}$ such that
\begin{eqnarray} \label{BKPtr}
g \phi_n g^{-1}=\sum_m B_{m n} \phi_m, \quad g \phi_n^* g^{-1}=\sum_m B^{'}_{n m} \phi_m^{*},
\end{eqnarray}
with
\begin{eqnarray}
\sum_l B_{l m} B^{'}_{n l}=\delta_{m n}.\nonumber
\end{eqnarray}
\end{itemize}
\end{lemma}

The key in this paper is to give the structure of the representations of the Clifford algebras $\mathcal{A}$ and $\mathcal{A}_{B}$. Following the notations of the vacuum states $|0\rangle$ and $\langle 0|$ are given by \cite{DateN}, the left module $\mathcal{F}$ of $\mathcal{A}$ is defined by generators $|0\rangle$ with relations
\begin{eqnarray}
\psi_{i}|0 \rangle =0\;(i<0),\;\;\psi_{i}^{*}|0\rangle=0\;(i\geq0), \label{leftmo}
\end{eqnarray}
the right module $\mathcal{F}^{*}$ consist of vectors $\langle 0|$ and relations
\begin{eqnarray}
 \langle 0 | \psi_{i} =0\;(i\geq0),\;\;\langle 0|\psi_{i}^{*}=0\;(i<0). \label{rightmo}
\end{eqnarray}
A pairing $\mathcal{F}^{*} \times \mathcal{F} \rightarrow \mathbb{C}: \langle 0 |a \otimes b |0 \rangle \mapsto \langle ab \rangle$ can be defined by \eqref{com-r}, \eqref{leftmo}, \eqref{rightmo} and
\begin{eqnarray}
\langle 0 |1 |0 \rangle=1.
\end{eqnarray}

The representations of $\mathcal{A}_B$ can be given by Kac and van de Leur \cite{KacTg} in the similar way. The left $\mathcal{A}_B$-module $\tilde{\mathcal{F}}$ with the vacuum and the right $\mathcal{A}_B$-module $\tilde{\mathcal{F}}^{*}$ with its dual can be constructed as
\begin{eqnarray*}
&\phi_n|0\rangle=0(n<0), & \langle 0| \phi_n=0(n>0),  \nonumber\\
&\phi_0|0\rangle=\frac{1}{\sqrt{2}}|1\rangle,  & \langle 0| \phi_0=\langle 1| \frac{1}{\sqrt{2}}.
\end{eqnarray*}

Let
\begin{eqnarray}\label{Fer-gen-fun}
\psi(\lambda)&=&\Sigma_{n \in \mathbf{Z}} \psi_{n} \lambda^{n},\nonumber\\
\psi^{*}(\lambda)&=&\Sigma_{n \in \mathbf{Z}} \psi_{n}^{*} \lambda^{-n},\nonumber\\
H(t) &=& \sum_{l=1}^{\infty} \sum_{n \in \mathbb{Z}}\psi_{n}\psi_{n+l}^{*}t_{l},\nonumber
\end{eqnarray}
and the Fermionic generating function $\phi(\lambda)$, the Hamiltonian $H_B(\tilde{t})$ defined by
\begin{eqnarray*}
\phi(\lambda)=\sum_{m \in \mathbb{Z}} \phi_m \lambda^m, \qquad H_B(\tilde{t})=\frac{1}{2}\sum_{n=0}^{\infty} \sum_{m \in \mathbb{Z}}(-1)^{m+1} \phi_m \phi_{-m-2 n-1} t_{2 n+1},
\end{eqnarray*}
where $\tilde{t}=\left(t_1, t_3, t_5, \cdots\right)$.

By \eqref{com-r}, one can obtain the commutator relations in the following lemma.
\begin{lemma}
\begin{eqnarray}
&\left[ \psi(p), \psi^{*}(q)\right]_{+}=q\delta \left(q-p\right),\label{phi-com-r}\\
&\left[\phi\left(p\right), \phi\left(q\right)\right]_{+}=q\delta \left(q-p\right).
\end{eqnarray}
where the formal $\delta$-function is defined by
\begin{eqnarray}
\delta \left(q-p\right)=q^{-1}\sum_{n\in \mathbb{Z}} \left(\frac{p}{q}\right)^{n}.
\end{eqnarray}
\end{lemma}

Additionally, $\langle l|=\langle 0|\Psi_{l}^{*}, \;|l \rangle =\Psi_{l}|0\rangle, l\in \mathbb{Z}$,
where
\begin{eqnarray}
\Psi_{l}^{*}=\left\{\begin{array}{ll}
\psi_{-1} \cdots \psi_{l}, & l<0 \\
1, & l=0\\
\psi_{0}^{*} \cdots \psi_{l-1}^{*}, & l>0
\end{array}\right. \nonumber,
\end{eqnarray}
and
\begin{eqnarray}
\Psi_{l}=\left\{\begin{array}{ll}
\psi_{l}^{*}\cdots \psi_{-1}^{*}, & l<0 \\
1, & l=0\\
\psi_{l-1}\cdots \psi_{0}, & l>0
\end{array}\right. \nonumber.
\end{eqnarray}

We list the following useful formula in the following lemma given by \cite{JimboS,DateN}.
\begin{lemma} The following identities hold:
\begin{itemize}
\item[\rm(1).]\begin{eqnarray}\label{eH}
e^{H(t)}\psi(\lambda)e^{-H(t)}=e^{\xi (t, \lambda)}\psi(\lambda),\;\;e^{H(t)}\psi^{*}(\lambda)e^{-H(t)}=e^{-\xi (t, \lambda)}\psi^{*}(\lambda),
\end{eqnarray}
\begin{eqnarray}\label{forim}
\langle l|\psi(\lambda)=\lambda^{l-1}\langle l-1| e^{H\left(-\epsilon (\lambda)\right)},\;\;\langle l|\psi^{*}(\lambda)=\lambda^{-l}\langle l+1|e^{H\left(\epsilon (\lambda)\right)},\;\;l \in \mathbb{Z},
\end{eqnarray}
where
$
\xi (t, \lambda)=\sum_{n=1}^{\infty} t_{n}\lambda^{n}, \epsilon (\lambda)=\left( \frac{1}{\lambda}, \frac{1}{2\lambda^{2}}, \cdots, \frac{1}{n\lambda^{n}}, \cdots\right).
$

\item[\rm(2).] \begin{eqnarray}\label{BKPusefor}
e^{H_{B}(\tilde{t})} \phi_n(\lambda) e^{-H_{B}(\tilde{t})}=e^{\tilde{\xi}(\tilde{t}, \lambda)}\phi_n(\lambda), e^{H_{B}(\tilde{t})} \phi_n e^{-H_{B}(\tilde{t})}=\sum_{l=0}^{\infty} \phi_{n-l} \tilde{p}_l(\tilde{t}),\label{BKPphi}
\end{eqnarray}
\begin{eqnarray}
\langle 0|\phi(\lambda)=\frac{1}{\sqrt{2}}\langle 1| e^{H_B\left(-\tilde{\epsilon} (\lambda)\right)}, \langle 1|\phi(\lambda)=\frac{1}{\sqrt{2}}\langle 0| e^{H_B\left(-\tilde{\epsilon} (\lambda)\right)},
\end{eqnarray}
where $\tilde{\xi}(\tilde{t}, \lambda)=\sum_{n=1}^{\infty} t_{2n-1}\lambda^{2n-1}, \tilde{\epsilon} (\lambda)=\left( \frac{1}{\lambda}, \frac{1}{3\lambda^{3}}, \cdots, \frac{1}{2n-1\lambda^{2n-1}}, \cdots\right)$.

\end{itemize}

\end{lemma}

One shall also need the following vertex operators,
\begin{eqnarray}
X(\lambda)&=&e^{\xi(t, \lambda)} e^{-\xi(\tilde{\partial}, \lambda^{-1})}=\sum_{i \in \mathbf{Z}} X_{i} (t, \tilde{\partial}  ) \lambda^{i},\nonumber\\
X^{*}(\lambda)&=&e^{-\xi(t, \lambda)} e^{\xi(\tilde{\partial}, \lambda^{-1})}=\sum_{i \in \mathbf{Z}} X_{i}^{*} (t, \tilde{\partial} ) \lambda^{-i},\nonumber\\
X_B(\tilde{t}, \lambda)&=&e^{\tilde{\xi}(\tilde{t}, \lambda)} e^{-2 \tilde{\xi}\left(\tilde{\partial}_{B}, \lambda^{-1}\right)}=\sum_{i \in \mathbb{Z}} X_{B, i} \lambda^i,
\end{eqnarray}
where $\tilde{\partial}=\left(\frac{\partial}{\partial_{t_{1}}}, \frac{1}{2} \frac{\partial}{\partial t_{2}}, \frac{1}{3} \frac{\partial}{\partial t_{3}}, \cdots\right)
$, $\tilde{\partial}_{B}=\left(\frac{\partial}{\partial_{t_{1}}}, \frac{1}{3} \frac{\partial}{\partial t_{3}}, \frac{1}{5} \frac{\partial}{\partial t_{5}}, \cdots\right)$ and $X_{i} (t, \tilde{\partial}  ) , X_{i}^{*} (t, \tilde{\partial} ), X_{B, i}$ are given by
\begin{eqnarray}
X_{i}(t, \tilde{\partial} )&=&\sum_{n = 0}^{\infty} p_{n+i}(t) p_{n}(-\tilde{\partial}),\nonumber\\
X_{i}^{*}(t, \tilde{\partial}  )&=&\sum_{n = 0}^{\infty} p_{n-i}(-t) p_{n}(\tilde{\partial}),\nonumber\\
X_{B, i}(\tilde{t}, \tilde{\partial}_B)&=&\sum_{n = 0}^{\infty} \tilde{p}_{n+i}(\tilde{t}) \tilde{p}_{n}(-2\tilde{\partial}_B)\nonumber
\end{eqnarray}
with the Schur polynomials $p_{n}(t)$ and $\tilde{p}_{n}(\tilde{t})$ satisfying
$
e^{\xi(t, \lambda)}=\sum_{n=0}^{\infty}p_{n}(t)\lambda^{n}
$
and $e^{\tilde{\xi}(\tilde{t}, \lambda)}=\sum_{n=0}^{\infty}\tilde{p}_{n}(\tilde{t})\lambda^{n}$.

\subsection{The $c$-$k$ constrained KP hierarchy}
\addcontentsline{toc}{section}{2.2 The $c$-$k$ constrained KP hierarchy}
Here we introduce the KP and $c$-$k$ constrained KP hierarchies \cite{DKS,KonopelchenkoD,ChengT}. In the previous section, the KP hierarchy is defined by \eqref{KPLax}. By introducing the dressing operator
\begin{eqnarray}
\phi=1+\sum\limits_{i=1}^{\infty} b_{i}\partial^{-i},\nonumber
\end{eqnarray}
which satisfies
\begin{eqnarray}
L=\phi\partial\phi^{-1},\nonumber
\end{eqnarray}
one can obtain that the KP hierarchy is equivalent to Sato equation
\begin{eqnarray}
\frac{\partial \phi}{\partial t_{n}}=-\left(L^{n}\right)_{-}\phi.\nonumber
\end{eqnarray}
The wave function $w(t,\lambda)$ of the KP hierarchy is defined by
\begin{eqnarray}
w(t,\lambda)=\phi\exp\left(\xi(t,\lambda)\right).\nonumber
\end{eqnarray}
The wave function satisfies the equations
\begin{eqnarray}
Lw&=&\lambda w,\nonumber\\
\frac{\partial w}{\partial t_{n}}&=&B_{n}w,\nonumber
\end{eqnarray}
where $n\in\mathbb{Z}_{\geq1}$.
While the adjoint wave function $w^{*}(t,\lambda)$                                                                                                                                                                                                                                                                                                                                                                                                     is defined by
\begin{eqnarray}
w^{*}(t,\lambda)=(\phi^{*})^{-1}\exp\left(-\xi(t,\lambda)\right).\nonumber
\end{eqnarray}
The adjoint wave function satisfies the equations
\begin{eqnarray}
L^{*}w^{*}&=&\lambda w^{*},\nonumber\\
\frac{\partial w^{*}}{\partial t_{n}}&=&-B_{n}^{*}w^{*},\nonumber
\end{eqnarray}
where $n\in\mathbb{Z}_{\geq1}$. Then the bilinear identity of the KP hierarchy is given as
\begin{eqnarray}\label{KPbil}
\operatorname{Res}_\lambda w(t, \lambda) w^{*}(t^{\prime}, \lambda)=0.
\end{eqnarray}

According to Sato theory, the tau function of the KP hierarchy can be defined as
\begin{eqnarray}\label{KPtau}
\tau(t)=\langle 0| e^{H(t)}g|0 \rangle.
\end{eqnarray}
Then the wave function $w(t, \lambda)$ and the dual wave function $w^{*}(t,\lambda)$ are given by the following formulas,
\begin{eqnarray}
&\;&w(t, \lambda)=\frac{\langle 1| e^{H(t)}\psi(\lambda)g|0 \rangle}{\langle 0| e^{H(t)}g|0 \rangle}=\frac{\tau(t-\epsilon (\lambda))}{\tau(t)} e^{\xi(t,\lambda)},\nonumber\\
&\;&w^{*}(t, \lambda)=\frac{\langle -1| e^{H(t)}\psi^{*}(\lambda)g|0 \rangle}{\langle 0| e^{H(t)}g|0 \rangle}=\frac{\tau(t+\epsilon (\lambda))}{\tau(t)} e^{-\xi(t,\lambda)}.\nonumber
\end{eqnarray}
The bilinear identity \eqref{KPbil} can be rewritten as
\begin{eqnarray}\label{KPbiltau}
\operatorname{Res}_\lambda \tau(t-\epsilon (\lambda))\tau(t^{\prime}+\epsilon (\lambda))e^{\xi(t-t^{\prime},\lambda)}=0.
\end{eqnarray}
The correspondence between $\mathcal{F}$ and $\mathcal{B}=\mathbb{C}[t_{1}, t_{2}, \cdots]$ is given by \cite{JimboS}.
\begin{lemma}\label{lemmaThecorresp}
The correspondence
\begin{eqnarray}
\nu: \mathcal{F}&\rightarrow&\mathcal{B}\nonumber\\
|u\rangle&\mapsto& \oplus \langle l|e^{H(t)}|u \rangle\nonumber
\end{eqnarray}
is an isomorphism.
\end{lemma}

In \cite{LorisKs}, the $c$-$k$ constrained KP hierarchy is defined by
\begin{subequations}\label{ckKPeq}
\begin{alignat}{2}
&L_{t_{n}}&=&\left[B_{n}, L\right],  \\
&L^{k}&=&B_{k}+\sum_{i=1}^{N} q_{i} \partial^{-1} r_{i}-cL^{-1}, \quad N \geqslant 1, \label{ckKPLax}\\
&q_{i,t_{n}}&=&B_{n} q_{i}, \\
&r_{i,t_{n}}&=&-B_{n}^{*} r_{i}.
\end{alignat}
\end{subequations}
By introducing
\begin{eqnarray}
q_{i}(t)=\frac{\rho_{i}(t)}{\tau(t)},\;\;i=1, 2, \cdots, N,\nonumber\\
r_{i}(t)=\frac{\sigma_{i}(t)}{\tau(t)},\;\;i=1, 2, \cdots, N,\nonumber
\end{eqnarray}
the bilinear formulation of the generalized  $k$-constrained KP hierarchy is obtained as
\begin{subequations}\label{ckBE}
\begin{alignat}{2}
&\sum_{i=1}^{N} \rho_{i}(t) \sigma_{i}\left(t^{\prime}\right)=\operatorname{Res}_{\lambda}\left(\left(\lambda^{k}+c \lambda^{-1}\right) \tau(t-\epsilon(\lambda)) \tau\left(t^{\prime}+\epsilon(\lambda)\right) e^{\xi\left(t-t^{\prime}, \lambda\right)}\right), \label{ckBE1}\\
&\rho_{i}(t) \tau\left(t^{\prime}\right)=\operatorname{Res}_{\lambda}\left(\lambda^{-1} \tau(t-\epsilon(\lambda)) \rho_{i}\left(t^{\prime}+\epsilon(\lambda)\right) e^{\xi\left(t-t^{\prime}, \lambda\right)}\right), i=1,2, \ldots, N,\label{ckBE2}\\
&\sigma_{i}\left(t^{\prime}\right) \tau(t)=\operatorname{Res}_{\lambda}\left(\lambda^{-1} \sigma_{i}(t-\epsilon(\lambda)) \tau\left(t^{\prime}+\epsilon(\lambda)\right) e^{\xi\left(t-t^{\prime}, \lambda\right)}\right), i=1,2, \ldots, N.\label{ckBE3}
\end{alignat}
\end{subequations}
From \eqref{ckBE}, an alternative bilinear form \cite{LorisKs} for the $c$-$k$ constraint KP hierarchy is given by
\begin{eqnarray}\label{ckbi}
&\operatorname{Res}_{\lambda}\left(\left(\tau_{t_{k}}(t-\epsilon(\lambda))+c\left(x-\lambda^{-1}\right) \tau(t-\epsilon(\lambda))\right)\right.\nonumber\\
&\left.\times\left(\tau_{t_{k}}\left(t^{\prime}+\epsilon(\lambda)\right)+c\left(x^{\prime}+\lambda^{-1}\right) \tau\left(t^{\prime}+\epsilon(\lambda)\right)\right) \exp \xi\left(t-t^{\prime}, \lambda\right)\right)=0
\end{eqnarray}
The simplest non-trivial Hirota bilinear equation contained in \eqref{ckbi} is
\begin{eqnarray}\label{simplebi}
\left(4 D_{1} D_{3} D_{k}^{2}-3 D_{2}^{2} D_{k}^{2}-D_{1}^{4} D_{k}^{2}+16 c D_{3} D_{k}-16 c D_{1}^{3} D_{k}-48 c^{2} D_{1}^{2}\right) \tau \cdot \tau=0 .
\end{eqnarray}
Let $k=1$ and $v=\log \tau$, the equation \eqref{simplebi} reads as
\begin{eqnarray}
v_{2 t_{2}}-4 c v_{2 x}-v_{4 x}-2 v_{2 x}^{2}+\frac{v_{3 x}^{2}-v_{x, t_{2}}^{2}}{v_{2 x}+4 c}=0, \nonumber
\end{eqnarray}
which is called the non-local Boussinesq equation in \cite{WillcoxBo}.

\bigskip

\subsection{The $c$-$k$ constrained BKP hierarchy}
\addcontentsline{toc}{section}{2.2 The $c$-$k$ constrained BKP hierarchy}
We recall some properties of the BKP hierarchy\cite{DateN,ChengC,LorisSr}. In the previous section we showed that the BKP hierarchy is a sub-hierarchy of the KP hierarchy and the Lax operator of the BKP hierarchy satisfies \eqref{BKPL}, thereafter we denote by $\tilde{L}$ the Lax operator satisfying the condition \eqref{BKPL}. The BKP hierarchy can be equivalently defined as
\begin{eqnarray}\label{BKPwavefun}
&\;&\tilde{L} \tilde{w}(\tilde{t}, \lambda)=\lambda \tilde{w}(\tilde{t}, \lambda), \quad \tilde{t}=\left(t_1, t_3, t_5, \cdots\right) \\
&\;&\frac{\partial \tilde{w}(\tilde{t}, \lambda)}{\partial t_{2 n+1}}=\tilde{B}_{2 n+1} \tilde{w}(\tilde{t}, \lambda), \quad n=0,1,2, \cdots,
\end{eqnarray}
where $\tilde{B}_{2 n+1}=(\tilde{L}^{2 n+1})_{+}$.
The solution $\tilde{w}(\tilde{t}, \lambda)$ satisfying \eqref{BKPwavefun} is called the wave function of the BKP hierarchy. The wave function of the BKP hierarchy satisfies the bilinear identity
\begin{eqnarray}\label{BKPBi}
\operatorname{Res}_\lambda \lambda^{-1} \tilde{w}(\tilde{t}, \lambda)\tilde{w}\left(\tilde{t}^{\prime},-\lambda\right)=1.
\end{eqnarray}
Let $\tilde{\phi}=1+\sum_{i=1}^{\infty} \tilde{w}_i \partial^{-i}$ be the wave operator of the BKP hierarchy. Then the Lax operator $\tilde{L}$ and the wave function $\tilde{w}(\tilde{t}, \lambda)$ of the BKP hierarchy satisfy
\begin{eqnarray}
\tilde{L}=\tilde{\phi} \partial \tilde{\phi}^{-1}, \quad \tilde{w}(\tilde{t}, \lambda)=\tilde{\phi} e^{\tilde{\xi}(\tilde{t}, \lambda)}=\widehat{w} e^{\tilde{\xi}(\tilde{t}, \lambda)},\nonumber
\end{eqnarray}
where
\begin{eqnarray}
\widehat{w}=1+\frac{w_1}{\lambda}+\frac{w_2}{\lambda^2}+\frac{w_3}{\lambda^3}+\cdots.\nonumber
\end{eqnarray}
The adjoint wave function $\tilde{w}^*(\tilde{t}, \lambda)$ of the BKP hierarchy is
\begin{eqnarray}
\tilde{w}^*(\tilde{t}, \lambda)=(\tilde{\phi}^*)^{-1} e^{-\tilde{\xi}(\tilde{t}, \lambda)}=-\lambda^{-1} \tilde{w}_x(\tilde{t},-\lambda),\nonumber
\end{eqnarray}
which satisfies
\begin{eqnarray}
\tilde{L}^* \tilde{w}^*(t, \lambda)&=&\lambda \tilde{w}^*(t, \lambda), \quad \nonumber\\
\frac{\partial \tilde{w}^*(\tilde{t}, \lambda)}{\partial t_{2 n+1}}&=&-\tilde{B}_{2 n+1}^* \tilde{w}^*(t, \lambda), \quad n=0,1,2, \cdots.\nonumber
\end{eqnarray}

The BKP hierarchy has a separate tau function, which is defined as
\begin{eqnarray}\label{BKPtau}
\tilde{\tau}(\tilde{t})=\langle 0| e^{H_{B}(\tilde{t})}\tilde{g}|0 \rangle=\langle 1| e^{H_{B}(\tilde{t})}\tilde{g}|1 \rangle,
\end{eqnarray}
where $\tilde{g} \in B(\infty)$.
The BKP hierarchy can be described by a tau function as
\begin{eqnarray}
\tilde{w}(\tilde{t}, \lambda)=\frac{2 \langle 0|\phi_{0} e^{H_{B}(\tilde{t})}\phi(\lambda)\tilde{g}|0 \rangle}{\langle 0| e^{H_{B}(\tilde{t})}\tilde{g}|0 \rangle}=\frac{\tilde{\tau}\left(\tilde{t}-2\tilde{\epsilon} (\lambda)\right)}{\tilde{\tau}(\tilde{t})} e^{\tilde{\xi}(\tilde{t}, \lambda)},\nonumber
\end{eqnarray}
where $\tilde{\epsilon} (\lambda)=\left(\frac{1}{\lambda}, \frac{1}{3\lambda^3}, \frac{1}{5\lambda^5}, \cdots\right)$. Thus, the bilinear identity \eqref{BKPBi} becomes
\begin{eqnarray}\label{BKPBi2}
\operatorname{Res}_\lambda \lambda^{-1} \tilde{\tau}\left(\tilde{t}-2\left[\lambda^{-1}\right]\right) \tilde{\tau}\left(\tilde{t}^{\prime}+2\left[\lambda^{-1}\right]\right) e^{\tilde{\xi}\left(\tilde{t}-\tilde{t}^{\prime}, \lambda\right)}=\tilde{\tau}(\tilde{t}) \tilde{\tau}\left(\tilde{t}^{\prime}\right).
\end{eqnarray}
The correspondence between $\tilde{\mathcal{F}}$ and $\tilde{\mathcal{B}}=\mathbb{C}[t_{1}, t_{3}, t_{5}, \cdots]$ is given by \cite{JimboS}.

\begin{lemma}\label{lemmaThecorrespB}
The correspondence
\begin{eqnarray}
\tilde{\nu}: \tilde{\mathcal{F}}&\rightarrow&\tilde{\mathcal{B}}\nonumber\\
|\tilde{u}\rangle&\mapsto& \langle 0|e^{H_{B}(\tilde{t})}|\tilde{u} \rangle \oplus \langle 1|e^{H_{B}(\tilde{t})}|\tilde{u} \rangle \nonumber
\end{eqnarray}
is an isomorphism.
\end{lemma}

Now we define the $c$-$k$ constrained BKP hierarchy, which is given by
\begin{subequations}\label{ckBKPeq}
\begin{alignat}{2}
&\frac{\partial \tilde{L}}{\partial \tilde{t}_{2 m+1}}&=&\left[\tilde{B}_{2 m+1}, \tilde{L}\right], \\
&\tilde{L}^{2 k+1}&=&\tilde{B}_{2 k+1}+\sum_{i=1}^{\tilde{N}}\left(\tilde{r}_i \partial^{-1} \tilde{q}_{i, x}-\tilde{q}_i \partial^{-1} \tilde{r}_{i, x}\right)-c \tilde{L}^{-1}, \quad k=0,1,2, \cdots,\label{ckBKPeq2}\\
&\frac{\partial \tilde{q}_i(\tilde{t})}{\partial \tilde{t}_{2 m+1}}&=&\tilde{B}_{2 m+1} \tilde{q}_i(\tilde{t}), \\
&\frac{\partial \tilde{r}_i(\tilde{t})}{\partial \tilde{t}_{2 m+1}}& =&\tilde{B}_{2 m+1} \tilde{r}_i(\tilde{t}).
\end{alignat}
\end{subequations}

Here two important lemmas are given in \cite{DKS,ShenOa}.

\begin{lemma}\label{PQclass}
For two pseudo-differential operator $P$ and $Q$,
\begin{eqnarray}
\operatorname{R es}_\lambda\left[\left(P e^{x \lambda}\right)\left(Q e^{-x \lambda}\right)\right]=\operatorname{Res}_{\partial} P Q^*\nonumber
\end{eqnarray}
holds, where $Q^*$ is the formal adjoint of $Q$.
\end{lemma}

\begin{lemma}\label{cBKPrq}
By introducing
\begin{eqnarray}
\tilde{r}_i(t)=\frac{\tilde{\rho}_i(\tilde{t})}{\tilde{\tau}(\tilde{t})}, \quad \tilde{q}_i(\tilde{t})=\frac{\tilde{\sigma}_i(\tilde{t})}{\tilde{\tau}(\tilde{t})}, \quad i=1,2, \cdots, m,\nonumber
\end{eqnarray}
then the functions $\Omega_i$ and $\Gamma_i$ defined by $\Omega_{i, x}=\tilde{q}_i \psi^*$ and $\Gamma_{i, x}=\tilde{r}_i \psi^*$ respectively can also be rewritten as
\begin{subequations}
\begin{alignat}{2}
\Omega_i(\tilde{t}, \lambda)&=-\frac{\tilde{\sigma}_i(\tilde{t}) \tilde{\tau}\left(\tilde{t}+2\tilde{\epsilon}(\lambda)\right)+\tilde{\sigma}_i\left(\tilde{t}+2\tilde{\epsilon}(\lambda)\right) \tilde{\tau}(\tilde{t})}{2 \lambda \tilde{\tau}^2(\tilde{t})} e^{-\tilde{\xi}(\tilde{t}, \lambda)} \nonumber,\\
\Gamma_i(\tilde{t}, \lambda)&=-\frac{\tilde{\rho}_i(\tilde{t}) \tilde{\tau}\left(\tilde{t}+2\tilde{\epsilon}(\lambda)\right)+\tilde{\rho}_i\left(\tilde{t}+2\tilde{\epsilon}(\lambda)\right) \tilde{\tau}(\tilde{t})}{2 \lambda \tilde{\tau}^2(\tilde{t})} e^{-\tilde{\xi}(\tilde{t}, \lambda)} \nonumber.
\end{alignat}
\end{subequations}
\end{lemma}

\section{The Fermionic picture }

In this section, we will give the Fermionic picture of the $c$-$k$ constrained KP and BKP hierarchies by using the Clifford algebras $\mathcal{A}$ and $\mathcal{A}_B$ defined in Definition \ref{def:Clifford}.

\begin{thm}\label{thmFermionic}
The Fermionic picture of the $c$-$k$ constrained KP and BKP hierarchies can be given as follows.

{\rm \textbf{(\Rmnum{1}).}}For the $c$-$k$ constrained KP hierarchy, if denote $|u \rangle=\nu^{-1}(\tau(t)) \in  G(\infty)|l\rangle$, $|f_{i} \rangle=\nu^{-1}(\rho_{i}(t)) \in  G(\infty)|s\rangle$ and $|g_{i} \rangle=\nu^{-1}(\sigma_{i}(t)) \in  G(\infty)|n\rangle$ with $i=1, 2, \cdots, N$, then
\begin{subequations}\label{thmFer}
\begin{alignat}{2}
&\sum_{m\in \mathbb{Z}} \psi_{m}|u\rangle \otimes \psi_{m}^{*}\left|u\right\rangle=0,\label{thmFer0}\\
&\sum_{m\in \mathbb{Z}} \psi_{l-s+1+m}|u\rangle \otimes \psi_{m}^{*}\left|f_{i}\right\rangle=\left|f_{i}\right\rangle \otimes|u\rangle, \label{thmFera}\\
&\sum_{m\in \mathbb{Z}} \psi_{n-l+1+m}|g_{i}\rangle \otimes \psi_{m}^{*}\left|u\right\rangle=\left|u\right\rangle \otimes|g_{i}\rangle,\label{thmFerb}\\
&\sum_{m\in \mathbb{Z}}\left(\psi_{m-k}|u\rangle \otimes \psi_{m}^{*}|u\rangle+\psi_{m+1}|u\rangle \otimes \psi_{m}^{*}|u\rangle\right)=\sum_{i=1}^{N}\left|f_{i}\right\rangle \otimes\left|g_{i}\right\rangle.\label{thmFerc}
\end{alignat}
\end{subequations}

{\rm \textbf{(\Rmnum{2}).}} For the $c$-$k$ constrained BKP hierarchy, if denote $|\tilde{u} \rangle=\tilde{\nu}^{-1}(\tilde{\tau}(\tilde{t})) \in  B(\infty)|0\rangle$, $|\tilde{f}_{i} \rangle=\tilde{\nu}^{-1}(\tilde{\rho}_{i}(\tilde{t})) \in  B(\infty)|1\rangle$ and $|\tilde{g}_{i} \rangle=\tilde{\nu}^{-1}(\tilde{\sigma}_{i}(\tilde{t})) \in  B(\infty)|1\rangle$ with $i=1, 2, \cdots, \tilde{N}$, then
\begin{subequations}\label{thmcBKPFer}
\begin{alignat}{2}
&\sum_{m\in \mathbb{Z}} (-1)^{m}\phi_{-m}|\tilde{u}\rangle \otimes \phi_{m}\left|\tilde{u}\right\rangle=\frac{1}{2}\left|\tilde{u}\right\rangle \otimes| \tilde{u}\rangle,\label{thmcBKPFer0}\\
&\sum_{m\in \mathbb{Z}} (-1)^{m}\phi_{-m}|\tilde{u}\rangle \otimes \phi_{m}\left|\tilde{g}_{i}\right\rangle=\left|\tilde{g}_{i}\right\rangle \otimes|\tilde{u}\rangle-\frac{1}{2}\left|\tilde{u}\right\rangle \otimes| \tilde{g}_{i}\rangle, \label{thmcBKPFera}\\
&\sum_{m\in \mathbb{Z}} (-1)^{m}\phi_{-m}|\tilde{u}\rangle \otimes \phi_{m}\left|\tilde{f}_{i}\right\rangle=\left|\tilde{f}_{i}\right\rangle \otimes|\tilde{u}\rangle-\frac{1}{2}\left|\tilde{u}\right\rangle \otimes| \tilde{f}_{i}\rangle, \label{thmcBKPFerb}\\
&\sum_{m\in \mathbb{Z}}(-1)^{m}\left(\phi_{m-k}|\tilde{u}\rangle+c\phi_{-m+1}|\tilde{u}\rangle \right)\otimes \phi_{m}|\tilde{u}\rangle=\sum_{i=1}^{\tilde{N}}\left(\left|\tilde{f}_{i}\right\rangle \otimes\left|\tilde{g}_{i}\right\rangle-\left|\tilde{g}_{i}\right\rangle \otimes\left|\tilde{f}_{i}\right\rangle\right).\label{thmcBKPFerc}
\end{alignat}
\end{subequations}
\end{thm}

\subsection{ The case of the $c$-$k$ constrained KP hierarchy}

We here prove Theorem \ref{thmFermionic} for the case {\rm \textbf{(\Rmnum{1})}} and give its interpretations. In addition to the knowledge in the previous section, we need an the following important lemmas in \cite{KacEh}.

\begin{lemma}\label{lemKPg}
If $g \in G(\infty)$ and $S=\sum_{m\in \mathcal{Z}}\psi_{m}\otimes \psi^{*}_{m}$, then
\begin{eqnarray}
\left[S, g\otimes g\right]=0.\nonumber
\end{eqnarray}
\end{lemma}

\begin{lemma}\label{lemgKPu}
If $|u\rangle \in \mathcal{F}$ and $|u\rangle \neq 0$, then $|u\rangle \in G(\infty)|0\rangle$ if and only if
\begin{eqnarray}
S(|u\rangle \otimes|u\rangle)=0.\nonumber
\end{eqnarray}
\end{lemma}

\noindent
\textit{Proof of Theorem \ref{thmFermionic} for the case {\rm \textbf{(\Rmnum{1})}}.}

By \eqref{BKPphi} and Lemma \ref{lemmaThecorresp}, we can know that
\begin{subequations}
\begin{alignat}{2}
\tau\left(t+\varepsilon(\lambda)\right) e^{-\xi\left(t, \lambda\right)}&=\lambda^{l-1} \nu(\psi^{*}(\lambda)|u\rangle), \nonumber\\
\tau\left(t-\varepsilon(\lambda)\right) e^{\xi\left(t, \lambda\right)}&=\lambda^{-l} \nu(\psi(\lambda)|u\rangle), \nonumber\\
\rho_{i}\left(t+\varepsilon(\lambda)\right) e^{-\xi\left(t, \lambda\right)}&=\lambda^{s-1} \nu(\psi^{*}(\lambda)|f_{i}\rangle), \nonumber\\
\sigma_{i}\left(t-\varepsilon(\lambda)\right) e^{\xi\left(t, \lambda\right)}&=\lambda^{-n} \nu(\psi(\lambda)|g_{i}\rangle). \nonumber
\end{alignat}
\end{subequations}
Then \eqref{thmFer0} is given by \eqref{KPbiltau} and \eqref{ckBE} can be rewritten as \eqref{thmFera}-\eqref{thmFerc}.
\hfill\qedsymbol

From Lemma \ref{lemKPg}, we can know that
\begin{proposition}
The equation \eqref{thmFer0} is equivalent to $|u \rangle \in G(\infty)$.
\end{proposition}

\begin{proposition}\label{proG0}
Let $|u \rangle \in G(\infty)|0 \rangle$, then the following are equivalent:

{\rm \textbf{(\Rmnum{1}).}}  The equation \eqref{thmFera} holds,

{\rm \textbf{(\Rmnum{2}).}}  there exists $\alpha_{i} \in V=\{\sum_{m\in \mathbb{Z}} c_{i, m}\psi_{m}\}$ such that $|f_{i}\rangle=\alpha_{i}|u \rangle$,

{\rm \textbf{(\Rmnum{3}).}} The following equation holds,
\begin{eqnarray}
\sum_{m\in \mathbb{Z}} \psi_{m}^{*}\left|f_{i}\right\rangle\otimes \psi_{m}|u\rangle=|u\rangle \otimes \left|f_{i}\right\rangle.\nonumber
\end{eqnarray}
\end{proposition}

\begin{proof}
Firstly, assume that $|u \rangle =|0\rangle$, we have $\left|f_{i}\right\rangle \in G|1 \rangle$, i.e., $s=1$. If \eqref{thmFera} holds, then
\begin{eqnarray}
\sum_{m\in \mathbb{Z}} \psi_{m}|0\rangle \otimes \psi_{m}^{*}\left|f_{i}\right\rangle=\left|f_{i}\right\rangle \otimes|0\rangle. \nonumber
\end{eqnarray}
Since all $\psi_{m}|0 \rangle (m\geq0)$ are independent, there exists $\alpha_{i} \in V$ such that $|f_{i}\rangle=\alpha_{i}|0 \rangle$.

By Lemma \ref{lemKPg}, when $|u \rangle =g|0\rangle (g \in GL(\infty))$, we can obtain that
\begin{eqnarray}
\sum_{m\in \mathbb{Z}} \psi_{m}|0\rangle \otimes \psi_{m}^{*}g^{-1}\left|f_{i}\right\rangle=g^{-1}\left|f_{i}\right\rangle \otimes|0\rangle. \nonumber
\end{eqnarray}
Therefore, there exists $\hat{\alpha}_{i} \in V$ such that $g^{-1}|f_{i}\rangle=\hat{\alpha}_{i}|0 \rangle$. {\rm \textbf{(\Rmnum{1})}} implies {\rm \textbf{(\Rmnum{2})}} by \eqref{KPtr}.

Secondly, if $|f_{i}\rangle=\alpha_{i}|u \rangle$ with $\alpha_{i}=\sum_{n\in \mathbb{Z}} c_{i, n}\psi_{n}$, then we only need to prove that by Lemma \ref{lemKPg} when $|u \rangle =|0\rangle$,
\begin{eqnarray}
&\;&\sum_{m\in \mathbb{Z}} \psi_{m}|0\rangle \otimes \psi_{m}^{*}\left|f_{i}\right\rangle=\sum_{m=0}^{\infty} \psi_{m}|0\rangle \otimes \psi_{m}^{*}\sum_{n=0}^{\infty} c_{i, n}\psi_{n}|0 \rangle\nonumber\\
&=&\sum_{m=0}^{\infty}\sum_{n=0}^{\infty}c_{i, n} \psi_{m}|0\rangle \otimes\left(\delta_{m,n}-\psi_{n} \psi_{m}^{*}\right)|0 \rangle=\sum_{m\in \mathbb{Z}} \psi_{m}^{*}\left|f_{i}\right\rangle \otimes \psi_{m}|0\rangle\nonumber.
\end{eqnarray}

Moreover, from the properties of the tensor product, we know that {\rm \textbf{(\Rmnum{1})}} and {\rm \textbf{(\Rmnum{3})}} are equivalent.
\end{proof}

Similarly, we can obtain the following proposition.

\begin{proposition}\label{proG1}
Let $|u \rangle \in G(\infty)|0 \rangle$, then the following are equivalent:

{\rm \textbf{(\Rmnum{1}).}}  The equation \eqref{thmFerb} holds,

{\rm \textbf{(\Rmnum{2}).}}  there exists $\beta_{i} \in V^{*}=\{\sum_{m\in \mathbb{Z}} d_{i, m}\psi_{m}^{*}\}$ such that $|g_{i}\rangle=\beta_{i}|u \rangle$,

{\rm \textbf{(\Rmnum{3}).}} The following equation holds,
\begin{eqnarray}
\sum_{m\in \mathbb{Z}} \psi_{m}^{*}\left|u\right\rangle \otimes\psi_{m}|g_{i}\rangle=|g_{i}\rangle \otimes \left|u\right\rangle.\nonumber
\end{eqnarray}
\end{proposition}

\begin{proposition}\label{proG2}
Let $|u \rangle \in G(\infty)|0 \rangle$, then the following are equivalent:

{\rm \textbf{(\Rmnum{1}).}}  The equation \eqref{thmFerc} holds,

{\rm \textbf{(\Rmnum{2}).}}  there exists $\alpha_{i}=\sum_{m\in \mathbb{Z}} c_{i, m}\psi_{m} $ and $\beta_{i}=\sum_{m\in \mathbb{Z}} d_{i, m}\psi_{m}^{*}$ such that
\begin{eqnarray}
\sum_{m \in \mathbb{Z}}(A^{-1})_{ma}\left(\psi_{m-k}+\psi_{m+1}\right)-\sum_{i=1}^{N}d_{i, m}\alpha_{i} \in Ann(|u\rangle), (a<0) \nonumber
\end{eqnarray}
or
\begin{eqnarray}
\sum_{m \in \mathbb{Z}}A_{am}\left(\psi^{*}_{m+k}+\psi^{*}_{m-1 }\right)-\sum_{i=1}^{N}c_{i, m}\beta_{i}\in Ann^{*}(|u\rangle), (a\geq0)  \nonumber
\end{eqnarray}
holds, where the matrix $A$ is given by \eqref{KPtr}, $Ann(|u \rangle)=\left \{ \alpha | \alpha \in V, \alpha |u\rangle =0 \right \}$ and $Ann^{*}(|u \rangle)=\left \{ \beta | \beta \in V^{*}, \beta |u\rangle =0 \right \}$,

{\rm \textbf{(\Rmnum{3}).}} The following equation holds,
\begin{eqnarray}
\sum_{m\in \mathbb{Z}}\left( \psi_{m}^{*}|u\rangle \otimes \psi_{m-k}|u\rangle+\psi_{m}^{*}|u\rangle \otimes \psi_{m+1}|u\rangle \right)=\sum_{i=1}^{N}\left|g_{i}\right\rangle \otimes \left|f_{i}\right\rangle .\nonumber
\end{eqnarray}
\end{proposition}

\begin{proof}
It is easy to know that {\rm \textbf{(\Rmnum{1})}} and {\rm \textbf{(\Rmnum{3})}} are equivalent.

On the one hand, by \eqref{KPtr} we have
\begin{eqnarray}
\sum_{m\in \mathbb{Z}}\left(\psi_{m-k}g|0\rangle+\psi_{m+1}g|0\rangle\right) \otimes \psi_{m}^{*}g|0\rangle=\sum_{a\in \mathbb{Z}_{<0}}\sum_{m \in \mathbb{Z}}(A^{-1})_{ma}\left(\psi_{m-k}+\psi_{m+1}\right)g|0\rangle \otimes g\psi_{a}^{*}|0\rangle\nonumber
\end{eqnarray}
and
\begin{eqnarray}
\sum_{m\in \mathbb{Z}}\left(\psi_{m-k}g|0\rangle+\psi_{m+1}g|0\rangle\right) \otimes \psi_{m}^{*}g|0\rangle=\sum_{a\in \mathbb{Z}_{\geq0}}g\psi_{a}|0\rangle \otimes \sum_{m \in \mathbb{Z}}A_{am}\left(\psi^{*}_{m+k}+\psi^{*}_{m-1}\right)g|0\rangle\nonumber
\end{eqnarray}
Since all $\psi_{m}|0 \rangle (m\geq0)$ and $\psi^{*}_{m}|0 \rangle (m<0)$ are independent, there exists $\alpha_{i}=\sum_{m\in \mathbb{Z}} c_{i, m}\psi_{m} $ and $\beta_{i}=\sum_{m\in \mathbb{Z}} d_{i, m}\psi_{m}^{*}$. If {\rm \textbf{(\Rmnum{1})}} holds, {\rm \textbf{(\Rmnum{2})}} also holds.

On the other hand, if
\begin{eqnarray}
\sum_{m \in \mathbb{Z}}(A^{-1})_{ma}\left(\psi_{m-k}+\psi_{m+1}\right) -\sum_{i=1}^{N}d_{i, m}\alpha_{i}\in Ann(|u\rangle), (a<0) \nonumber
\end{eqnarray}
holds, then
\begin{eqnarray}
&\;&\sum_{m\in \mathbb{Z}}\left(\psi_{m-k}g|0\rangle+\psi_{m+1}g|0\rangle\right) \otimes \psi_{m}^{*}g|0\rangle=\sum_{a\in \mathbb{Z}_{<0}}\sum_{m \in \mathbb{Z}}(A^{-1})_{ma}\left(\psi_{m-k}+\psi_{m+1}\right)g|0\rangle \otimes g\psi_{a}^{*}|0\rangle\nonumber\\
&=&\sum_{a\in \mathbb{Z}_{<0}}\sum_{m \in \mathbb{Z}}(A^{-1})_{ma}\sum_{i=1}^{N}d_{i, m}\alpha_{i} g|0\rangle\otimes g\psi_{a}^{*}|0\rangle=\sum_{i=1}^{N}\left|f_{i}\right\rangle \otimes\left|g_{i}\right\rangle\nonumber.
\end{eqnarray}
 if
\begin{eqnarray}
\sum_{m \in \mathbb{Z}}A_{am}\left(\psi^{*}_{m+k}+\psi^{*}_{m-1}\right)-\sum_{i=1}^{N}c_{i, m}\beta_{i} \in Ann^{*}(|u\rangle), (a\geq0)  \nonumber
\end{eqnarray}
holds, then {\rm \textbf{(\Rmnum{1})}} also holds.
\end{proof}

By Propositions \ref{proG0}-\ref{proG2}, we can obtain that
\begin{proposition}
The equation \eqref{ckBE} are equivalent to the following equations, respectively,
\begin{subequations}
\begin{alignat}{2}
&\sum_{i=1}^{N}\sigma_{i}(t)\rho_{i}\left(t^{\prime}\right)=\operatorname{Res}_{\lambda}\left(\left(\lambda^{k}+c \lambda^{-1}\right) \tau(t-\epsilon(\lambda)) \tau\left(t^{\prime}+\epsilon(\lambda)\right) e^{\xi\left(t-t^{\prime}, \lambda\right)}\right),\nonumber\\
&\tau_{i}(t) \rho\left(t^{\prime}\right)=\operatorname{Res}_{\lambda}\left(\lambda^{-1}\rho(t-\epsilon(\lambda)) \tau_{i}\left(t^{\prime}+\epsilon(\lambda)\right) e^{\xi\left(t-t^{\prime}, \lambda\right)}\right), i=1,2, \ldots, N,\nonumber\\
&\sigma_{i}(t)\tau\left(t^{\prime}\right)=\operatorname{Res}_{\lambda}\left(\lambda^{-1}\tau(t-\epsilon(\lambda)) \sigma_{i}\left(t^{\prime}+\epsilon(\lambda)\right) e^{\xi\left(t-t^{\prime}, \lambda\right)}\right), i=1,2, \ldots, N.\nonumber
\end{alignat}
\end{subequations}
\end{proposition}

\subsection{ The case of the $c$-$k$ constrained BKP hierarchy}

We here prove Theorem \ref{thmFermionic} for the case {\rm \textbf{(\Rmnum{2})}} and give its interpretations. We give several important results \cite{KacEh}.

\begin{lemma}\label{lemBKPg}
If $\tilde{g} \in B(\infty)$ and $\tilde{S}=\sum_{m\in \mathcal{Z}}(-1)^{m}\phi_{-m}\otimes \phi_{m}$, then
\begin{eqnarray}
\left[\tilde{S}, \tilde{g}\otimes \tilde{g}\right]=0.\nonumber
\end{eqnarray}
\end{lemma}

\begin{lemma}\label{lemgBKPu}
If $|\tilde{u}\rangle \in \mathcal{\tilde{F}}$ and $|\tilde{u}\rangle \neq 0$, then $|\tilde{u}\rangle \in B(\infty)|0\rangle$ if and only if
\begin{eqnarray}
\tilde{S}(|\tilde{u}\rangle \otimes|\tilde{u}\rangle)=\frac{1}{2}|\tilde{u}\rangle \otimes|\tilde{u}\rangle.\nonumber
\end{eqnarray}
\end{lemma}

\begin{proposition}
The eigenfunctions $\tilde{r}_i$ and $\tilde{q}_i$ of the $c$-$k$ construined BKP hierarchy satisfy the following relations
\begin{subequations}\label{cBKPeigenfun}
\begin{alignat}{2}
\sum_{i=1}^m\left(\tilde{q}_i(\tilde{t}) \tilde{r}_i\left(\tilde{t}^{\prime}\right)-\tilde{r}_i(\tilde{t}) \tilde{q}_i\left(\tilde{t}^{\prime}\right)\right)&=\operatorname{Res}_\lambda\left(\lambda^{k-1}+c \lambda^{-2}\right)\tilde{w}(\tilde{t}, \lambda)\tilde{w}^*(\tilde{t}^{\prime}, -\lambda) \label{cBKPeigenfun1}\\
\tilde{q}_i\left(\tilde{t}^{\prime}\right)&=-\operatorname{Res}_\lambda \tilde{w}\left(\tilde{t}^{\prime}, \lambda\right) \Omega_i(\tilde{t}, \lambda),\label{cBKPeigenfun2}\\
\tilde{r}_i\left(\tilde{t}^{\prime}\right)&=-\operatorname{Res}_\lambda \tilde{w}\left(\tilde{t}^{\prime}, \lambda\right) \Gamma_i(\tilde{t}, \lambda),\label{cBKPeigenfun3}
\end{alignat}
\end{subequations}
where $\Omega_{i, x}=\tilde{q}_i \tilde{w}^*$, and $\Gamma_{i, x}=\tilde{r}_i \tilde{w}^*$ for $i=1,2 \ldots, \tilde{N}$.
\end{proposition}

\begin{proof}
For $n \geq 0$, from \eqref{ckBKPeq2},
\begin{eqnarray}
\sum_{i=1}^m(-1)^n\left(\tilde{r}_i \partial^n \tilde{q}_{i, x}-\tilde{q}_i \partial^n \tilde{r}_{i, x}\right) &=&\operatorname{Res}_{\partial}\left[\left(\tilde{L}^k+c \tilde{L}^{-1}\right) \partial^n\right] \nonumber\\
&=&\operatorname{Res}_\lambda\left[\phi\left(\partial^k+c \partial^{-1}\right) e^{\tilde{\xi}(\tilde{t}, \lambda)}(-\partial)^n\left(\phi^{-1}\right)^* e^{-\tilde{\xi}(\tilde{t}, \lambda)}\right] \nonumber\\
&=&\operatorname{Res}_\lambda\left(\lambda^k+c \lambda^{-1}\right) \tilde{w}(\tilde{t}, \lambda) \tilde{w}^*\left(\tilde{t}^{\prime}, \lambda\right)\nonumber,
\end{eqnarray}
thus
\begin{eqnarray}\label{cBKPeigenfunadd}
\sum_{i=1}^m\left(\tilde{r}_i(\tilde{t}) \tilde{q}_{i, x^{\prime}}\left(\tilde{t}^{\prime}\right)-\tilde{q}_i(\tilde{t}) \tilde{r}_{i, x^{\prime}}\left(\tilde{t}^{\prime}\right)\right)=\operatorname{Res}_\lambda\left(\lambda^k+c \lambda^{-1}\right) \tilde{w}(\tilde{t}, \lambda) \tilde{w}^*\left(\tilde{t}^{\prime}, \lambda\right) .
\end{eqnarray}
Then we can obtain \eqref{cBKPeigenfun1} after integrating with respect to $x^{\prime}$, and the formula \eqref{cBKPeigenfunadd} can be rewritten as
\begin{eqnarray}
\sum_{i=1}^m\left(\tilde{r}_i(t) \tilde{q}_{i, x^{\prime}}\left(\tilde{t}^{\prime}\right)-\tilde{q}_i(\tilde{t}) \tilde{r}_{i, x^{\prime}}\left(\tilde{t}^{\prime}\right)\right) 
=\operatorname{Res}_\lambda \tilde{w}(\tilde{t}, \lambda) \sum_{i=1}^m\left(-\tilde{q}_{i, x^{\prime}} \Gamma_i+\tilde{r}_{i, x^{\prime}} \Omega_i\right).\nonumber
\end{eqnarray}
The equations \eqref{cBKPeigenfun2} and \eqref{cBKPeigenfun3} are true since $\tilde{r}_{i, x^{\prime}}$ and $\tilde{q}_{i, x^{\prime}}$ are the linearly independent functions.
\end{proof}

\begin{proposition}
The functions $\tilde{\sigma}_i(t), \tilde{\rho}_i(t)$, and $\tilde{\tau}_i(t)$ solve a solution for the $c$-$k$ constrained BKP hierarchy which satisfy the following bilinear equations
\begin{subequations}\label{cBKPfuntau}
\begin{alignat}{2}
\sum_{i=1}^m\left(\tilde{\sigma}_i(\tilde{t}) \tilde{\rho}_i\left(\tilde{t}^{\prime}\right)-\tilde{\rho}_i(\tilde{t}) \tilde{\sigma}_i\left(\tilde{t}^{\prime}\right)\right)&=\operatorname{Res}_\lambda\left(\lambda^{k-1}+c \lambda^{-2}\right) \tilde{\tau}\left(\tilde{t}-2\tilde{\epsilon}(\lambda)\right) \tilde{\tau}\left(\tilde{t}^{\prime}+2\tilde{\epsilon}(\lambda)\right) e^{\tilde{\xi}\left(\tilde{t}-\tilde{t}^{\prime}, \lambda\right)},\label{cBKPfuntau1}\\
2 \tilde{\sigma}_i(t) \tilde{\tau}\left(\tilde{t}^{\prime}\right)-\tilde{\tau}(\tilde{t})\tilde{\sigma}_i\left(\tilde{t}^{\prime}\right) &=\operatorname{Res}_\lambda \lambda^{-1} \tilde{\tau}\left(\tilde{t}-2\tilde{\epsilon}(\lambda)\right) \tilde{\sigma}_i\left(\tilde{t}^{\prime}+2\tilde{\epsilon}(\lambda)\right) e^{\tilde{\xi}\left(\tilde{t}-\tilde{t}^{\prime}, \lambda\right)},\label{cBKPfuntau2} \\
2 \tilde{\rho}_i(t) \tilde{\tau}\left(\tilde{t}^{\prime}\right)-\tilde{\tau}(\tilde{t})\tilde{\rho}_i\left(\tilde{t}^{\prime}\right) &=\operatorname{Res}_\lambda \lambda^{-1} \tilde{\tau}\left(\tilde{t}-2\tilde{\epsilon}(\lambda)\right) \tilde{\rho}_i\left(\tilde{t}^{\prime}+2\tilde{\epsilon}(\lambda)\right) e^{\tilde{\xi}\left(\tilde{t}-\tilde{t}^{\prime}, \lambda\right)}.\label{cBKPfuntau3}
\end{alignat}
\end{subequations}
\end{proposition}

\begin{proof}
By \eqref{cBKPeigenfun2} and Lemma \ref{cBKPrq}, we get
\begin{eqnarray}
2 \sigma_i\left(t^{\prime}\right) \tau^2(t)=\operatorname{Res}_\lambda \lambda^{-1} \tau\left(t^{\prime}-2\tilde{\epsilon}(\lambda)\right)\left(\sigma_i(t) \tau\left(t+2\tilde{\epsilon}(\lambda)\right)+\sigma_i\left(t+2\tilde{\epsilon}(\lambda)\right) \tau(t)\right) e^{-\bar{\xi}(t, \lambda)}.\nonumber
\end{eqnarray}
Taking account of the \eqref{BKPBi2} and substituting $t$ to $t^{\prime}$, we can rewrite the above formula as
\begin{eqnarray}
2 \sigma_i(t) \tau\left(t^{\prime}\right)-\sigma_i\left(t^{\prime}\right) \tau(t)=\operatorname{Res}_\lambda \lambda^{-1} \tau\left(t-2\tilde{\epsilon}(\lambda)\right) \sigma_i\left(t^{\prime}+2\tilde{\epsilon}(\lambda)\right) e^{\bar{\xi}\left(t-t^{\prime}, \lambda\right)}.\nonumber
\end{eqnarray}
The \eqref{cBKPfuntau2} is proved.
Similarly, we can prove \eqref{cBKPfuntau1} and \eqref{cBKPfuntau3}.
\end{proof}

\noindent
\textit{Proof of Theorem \ref{thmFermionic} for the case {\rm \textbf{(\Rmnum{2})}}.}

By \eqref{forim} and Lemma \ref{lemmaThecorrespB}, we can obtain that
\begin{subequations}
\begin{alignat}{2}
\tilde{\tau}\left(\tilde{t}-2\tilde{\epsilon}(\lambda\right) e^{\tilde{\xi}\left(\tilde{t}, \lambda\right)}&=\sqrt{2}\tilde{\nu}(\phi(\lambda)|\tilde{u}\rangle), \nonumber\\
\tilde{\tau}\left(\tilde{t}+2\tilde{\epsilon}(\lambda)\right) e^{-\tilde{\xi}\left(\tilde{t}, \lambda\right)}&=\sqrt{2}\tilde{\nu}(\phi(-\lambda)|\tilde{u}\rangle), \nonumber\\
\tilde{\rho}_{i}\left(\tilde{t}+2\tilde{\epsilon}(\lambda\right) e^{-\tilde{\xi}\left(\tilde{t}, \lambda\right)}&=\sqrt{2} \tilde{\nu}(\phi(-\lambda)|\tilde{f}_{i}\rangle), \nonumber\\
\tilde{\sigma}_{i}\left(\tilde{t}+2\tilde{\epsilon}(\lambda\right) e^{\tilde{\xi}\left(t, \lambda\right)}&=\sqrt{2} \tilde{\nu}(\phi(-\lambda)|\tilde{g}_{i}\rangle). \nonumber
\end{alignat}
\end{subequations}
The proof of the theorem is completed by using the above equation.
\hfill\qedsymbol

From Lemma \ref{lemgBKPu}, we can know that
\begin{proposition}
The equation \eqref{thmcBKPFer0} is equivalent to $|\tilde{u} \rangle \in B(\infty)$.
\end{proposition}

Similar to Proposition \ref{proG0}, we can obtain the following two propositions.

\begin{proposition}\label{procBKPG0}
Let $|\tilde{u} \rangle \in B(\infty)|0 \rangle$, then the following are equivalent:

{\rm \textbf{(\Rmnum{1}).}}  The equation \eqref{thmcBKPFera} holds,

{\rm \textbf{(\Rmnum{2}).}}  there exists $\tilde{\alpha}_{i} \in \tilde{V}=\{\sum_{m\in \mathbb{Z}} \tilde{c}_{i, m}\phi_{m}\}$ such that $|\tilde{g}_{i}\rangle=\tilde{\alpha}_{i}|\tilde{u} \rangle$,

{\rm \textbf{(\Rmnum{3}).}} The following equation holds,
\begin{eqnarray}
\sum_{m\in \mathbb{Z}} (-1)^{m} \phi_{m}\left|\tilde{g}_{i}\right\rangle\otimes\phi_{-m}|\tilde{u}\rangle=\left|\tilde{u}\right\rangle \otimes| \tilde{g}_{i}\rangle-\frac{1}{2}\left|\tilde{g}_{i}\right\rangle \otimes| \tilde{u}\rangle.\nonumber
\end{eqnarray}
\end{proposition}

\begin{proof}
We only need to show that the proposition holds when $|\tilde{u}\rangle =|0\rangle$ via Lemma \ref{lemBKPg}.

If \eqref{thmcBKPFera} holds, then
\begin{eqnarray}
\sum_{m\in \mathbb{Z}} (-1)^{m}\phi_{-m}|0\rangle \otimes \phi_{m}\left|\tilde{g}_{i}\right\rangle=\left|\tilde{g}_{i}\right\rangle \otimes|0\rangle-\frac{1}{2}\left|0\right\rangle \otimes| \tilde{g}_{i}\rangle. \nonumber
\end{eqnarray}
Since all $\phi_{i}|0 \rangle (i\geq0)$ are independent, there exists $\tilde{\alpha}_{i} \in \tilde{V}$ such that $|\tilde{g}_{i}\rangle=\tilde{\alpha}_{i}|0 \rangle$.

In turn, when $|\tilde{u}\rangle =|0\rangle$, there exists $\tilde{\alpha}_{i} \in \tilde{V}=\{\sum_{m\in \mathbb{Z}} \tilde{c}_{i, m}\phi_{m}\}$ such that $|\tilde{g}_{i}\rangle=\tilde{\alpha}_{i}|\tilde{u} \rangle$, so we have the left-hand side of \eqref{thmcBKPFera} is equal to
\begin{eqnarray}
\sum_{j\leq0}(-1)^j \phi_{-j}|0\rangle\otimes\sum_{m\geq0}\tilde{c}_{i, m}\left((-1)^i\delta_{j+m,0}-\phi_m\phi_j \right)=\sum_{j\geq0}\tilde{c}_{i, j}\phi_j|0\rangle\otimes|0\rangle-\frac{1}{2}|0\rangle\otimes\sum_{m\geq0}\tilde{c}_{i, m}\phi_{m}|0\rangle.\nonumber
\end{eqnarray}

Thus, {\rm \textbf{(\Rmnum{1})}} and {\rm \textbf{(\Rmnum{2})}} are equivalent.

From the properties of the tensor product, we know that {\rm \textbf{(\Rmnum{1})}} and {\rm \textbf{(\Rmnum{3})}} are equivalent.
\end{proof}

Similarly, we can obtain the following proposition.

\begin{proposition}\label{procBKPG1}
Let $|\tilde{u} \rangle \in B(\infty)|0 \rangle$, then the following are equivalent:

{\rm \textbf{(\Rmnum{1}).}}  The equation \eqref{thmcBKPFerb} holds,

{\rm \textbf{(\Rmnum{2}).}}  there exists $\tilde{\beta}_{i} \in \tilde{V}=\{\sum_{m\in \mathbb{Z}} \tilde{d}_{i, m}\phi_{m}\}$ such that $|\tilde{f}_{i}\rangle=\tilde{\beta}_{i}|\tilde{u} \rangle$,

{\rm \textbf{(\Rmnum{3}).}} The following equation holds,
\begin{eqnarray}
\sum_{m\in \mathbb{Z}} (-1)^{m} \phi_{m}\left|\tilde{f}_{i}\right\rangle\otimes\phi_{-m}|\tilde{u}\rangle=\left|\tilde{u}\right\rangle \otimes| \tilde{f}_{i}\rangle-\frac{1}{2}\left|\tilde{f}_{i}\right\rangle \otimes| \tilde{u}\rangle.\nonumber
\end{eqnarray}
\end{proposition}

Similar to Proposition \ref{proG2}, we can obtain the following proposition.

\begin{proposition}\label{procBKPG2}
Let $|\tilde{u} \rangle \in B(\infty)|0 \rangle$, then the following are equivalent:

{\rm \textbf{(\Rmnum{1}).}}  The equation \eqref{thmcBKPFerc} holds,

{\rm \textbf{(\Rmnum{2}).}}  there exists $\tilde{\alpha}_{i}=\sum_{m\in \mathbb{Z}} \tilde{c}_{i, m}\phi_{m} $ and $\tilde{\beta}_{i}=\sum_{m\in \mathbb{Z}} \tilde{d}_{i, m}\phi_{m}$ such that
\begin{eqnarray}
\sum_{j \in \mathbb{Z}}B_{-j,-m}\left(\phi_{-j-k}+c\phi_{-j+1}\right)-\sum_{i=1}^{\tilde{N}}\tilde{c}_{i, m}\tilde{\alpha}_{i}-\sum_{i=1}^{\tilde{N}}\tilde{d}_{i, m}\tilde{\beta}_{i} \in Ann(|\tilde{u}\rangle), (m\geq0) \nonumber
\end{eqnarray}
holds, where the matrix $B$ is given by \eqref{BKPtr}, $Ann(|\tilde{u} \rangle)=\left \{ \tilde{\alpha} | \tilde{\alpha} \in \tilde{V}, \tilde{\alpha} |\tilde{u}\rangle =0 \right \}$,

{\rm \textbf{(\Rmnum{3}).}} The following equation holds,
\begin{eqnarray}
\sum_{m\in \mathbb{Z}}(-1)^{m}\phi_{m}|\tilde{u}\rangle\otimes \left(\phi_{m-k}|\tilde{u}\rangle+c\phi_{-m+1}|\tilde{u}\rangle \right)=\sum_{i=1}^{\tilde{N}}\left( \left|\tilde{g}_{i}\right\rangle\otimes\left|\tilde{f}_{i}\right\rangle-\left|\tilde{f}_{i}\right\rangle \otimes\left|\tilde{g}_{i}\right\rangle \right).\nonumber
\end{eqnarray}
\end{proposition}

\begin{proof}
It is easy to know that {\rm \textbf{(\Rmnum{1})}} and {\rm \textbf{(\Rmnum{3})}} are equivalent.

Firstly, by \eqref{BKPtr}, we have the left-hand side of \eqref{thmcBKPFerc} is equal to
\begin{eqnarray}
&\;&\sum_{j}(-1)^{j}(\phi_{-j-k}+c\phi_{-j+1})g|0\rangle\otimes\sum_{m}(B^{-1})_{m,j}g\phi_m|0\rangle\nonumber\\
&=&\sum_{j,m}(-1)^mB_{-j,-m}(-1)^{j}(\phi_{-j-k}+c\phi_{-j+1})g|0\rangle\otimes g\phi_m|0\rangle\nonumber.
\end{eqnarray}
Since all $\phi_{m}|0 \rangle (m\geq0)$ are independent, there exists $\tilde{\alpha}_{i}=\sum_{m\in \mathbb{Z}} \tilde{c}_{i, m}\phi_{m} $ and $\tilde{\beta}_{i}=\sum_{m\in \mathbb{Z}} \tilde{d}_{i, m}\phi_{m}$ such that
\begin{eqnarray}
\sum_{j \in \mathbb{Z}}B_{-j,-m}\left(\phi_{-j-k}+c\phi_{-j+1}\right)-\sum_{i=1}^{\tilde{N}}\tilde{c}_{i, m}\tilde{\alpha}_{i}-\sum_{i=1}^{\tilde{N}}\tilde{d}_{i, m}\tilde{\beta}_{i} \in Ann(|\tilde{u}\rangle), (m\geq0) \nonumber
\end{eqnarray}
holds.

Conversely, if {\rm \textbf{(\Rmnum{2})}} holds, then {\rm \textbf{(\Rmnum{1})}} holds.
\end{proof}

By Propositions \ref{procBKPG0}-\ref{procBKPG2}, we can obtain that
\begin{proposition}
The equation \eqref{cBKPfuntau} are equivalent to the following equations, respectively,
\begin{subequations}\label{cBKPfuntau}
\begin{alignat}{2}
\sum_{i=1}^m\left(\tilde{\rho}_i(\tilde{t})\tilde{\sigma}_i\left(\tilde{t}^{\prime}\right)-\tilde{\sigma}_i(\tilde{t}) \tilde{\rho}_i\left(\tilde{t}^{\prime}\right)\right)&=\operatorname{Res}_\lambda\left(\lambda^{k-1}+c \lambda^{-2}\right) \tilde{\tau}\left(\tilde{t}-2\tilde{\epsilon}(\lambda)\right) \tilde{\tau}\left(\tilde{t}^{\prime}+2\tilde{\epsilon}(\lambda)\right) e^{\tilde{\xi}\left(\tilde{t}-\tilde{t}^{\prime}, \lambda\right)},\label{cBKPfuntau1}\\
2\tilde{\tau}(\tilde{t})  \tilde{\sigma}_i\left(\tilde{t}^{\prime}\right)-\tilde{\sigma}_i(\tilde{t})\tilde{\tau}\left(\tilde{t}^{\prime}\right) &=\operatorname{Res}_\lambda \lambda^{-1} \tilde{\sigma}_i\left(\tilde{t}-2\tilde{\epsilon}(\lambda)\right) \tilde{\tau}\left(\tilde{t}^{\prime}+2\tilde{\epsilon}(\lambda)\right) e^{\tilde{\xi}\left(\tilde{t}-\tilde{t}^{\prime}, \lambda\right)},\label{cBKPfuntau2} \\
2 \tilde{\tau}(\tilde{t}) \tilde{\rho}_i\left(\tilde{t}^{\prime}\right)-\tilde{\rho}_i(\tilde{t})\tilde{\tau}\left(\tilde{t}^{\prime}\right) &=\operatorname{Res}_\lambda \lambda^{-1} \tilde{\rho}_i\left(\tilde{t}-2\tilde{\epsilon}(\lambda)\right) \tilde{\tau}\left(\tilde{t}^{\prime}+2\tilde{\epsilon}(\lambda)\right) e^{\tilde{\xi}\left(\tilde{t}-\tilde{t}^{\prime}, \lambda\right)}.\label{cBKPfuntau3}
\end{alignat}
\end{subequations}
\end{proposition}

\section{Some solutions }

In this section, some solutions of the $c$-$k$ constrained KP and BKP hierarchies are given by using the Fermion operators.

\begin{thm}\label{thmBEsolution}
Some solutions of the $c$-$k$ constrained KP and BKP hierarchies can be given as follows.

{\rm \textbf{(\Rmnum{1}).}} For the $c$-$k$ constrained KP hierarchy, if $g \in GL(\infty)$ satisfies the condition
\begin{eqnarray}
g^{-1} \left(\Lambda_{k}+c \Lambda_{-1} \right)g=\sum_{m, n \in \mathbf{Z}} f_{m n} \psi_{m} \psi_{n}^{*} \nonumber
\end{eqnarray}
where $\Lambda_{k}=\Sigma_{n \in \mathrm{Z}} \psi_{n} \psi_{n+k}^{*} $, $\Lambda_{-1}=\Sigma_{n \in \mathrm{Z}} \psi_{n} \psi_{n-1}^{*} $ and $f_{m n}=\sum_{i=1}^{N} d_{m}^{(i)} e_{n}^{(i)}$ for $m \geqslant 0, n<0$, then
\begin{subequations}\label{ckBES}
\begin{alignat}{2}
&\tau(t)=\langle 0| e^{H(t)}g|0\rangle,\\
&\rho_{i}(t)=\sum_{m \geqslant 0} d_{m}^{(i)} \left\langle 1\left|e^{H(t)} g \psi_{m}\right| 0\right\rangle, \\
&\sigma_{i}(t)=\sum_{n<0} e_{n}^{(i)}\left\langle-1\left|e^{H(t)} g \psi_{n}^{*}\right| 0\right\rangle, i=1,2, \ldots, N
\end{alignat}
\end{subequations}
satisfy the bilinear equations of the $c$-$k$ constrained KP hierarchy.

{\rm \textbf{(\Rmnum{2}).}} For the $c$-$k$ constrained BKP hierarchy, let $\tilde{\Lambda}_k=\sum_{n \in \mathbb{Z}} \phi_n \phi_{n+k}^*$, if $\tilde{g} \in B(\infty)$ satisfies the condition
\begin{eqnarray}
\tilde{g}^{-1}\left(\tilde{\Lambda}_k+c \tilde{\Lambda}_{-1}\right) \tilde{g}=\sum_{l, j \in Z} \tilde{f}_{l j} \phi_l \phi_j^*, \quad \tilde{f}_{l j}=\sum_{i=1}^{\tilde{N}}\left(\tilde{d}_l^{(i)} \tilde{e}_j^{(i)}-(-1)^{l+j} \tilde{d}_{-j}^{(i)}\tilde{e}_{-l}^{(i)}\right),\nonumber
\end{eqnarray}
for $l \geq 0, j \leq 0$. Then $\tilde{\tau}(\tilde{t}), \tilde{\sigma}_i(\tilde{t})$ and $\tilde{\rho}_i(\tilde{t})$ solve the constrained BKP hierarchy, where
\begin{eqnarray}
\tilde{\tau}(\tilde{t})&=&\langle 0| e^{H_{B}(\tilde{t})}\tilde{g}|0 \rangle\nonumber\\
\tilde{\sigma}_i(\tilde{t})&=&2 \sum_{l \geq 0} \tilde{d}_i^{(i)}\left\langle 0\left|\phi_0 e^{H_B(\tilde{t})} \tilde{g} \phi_l\right| 0\right\rangle, \nonumber\\
\tilde{\rho}_i(\tilde{t})&=&2\sum_{j \leq 0} \tilde{e}_j^{(i)}\left\langle 0\left|\phi_0 e^{H_{B}(\tilde{t})} \tilde{g} \phi_j^*\right| 0\right\rangle, \quad i=1,2, \ldots, \tilde{N}.\nonumber
\end{eqnarray}
\end{thm}

\subsection{ The case of the $c$-$k$ constrained KP hierarchy}

We here prove Theorem \ref{thmBEsolution} for the case {\rm \textbf{(\Rmnum{1})}} and give several conclusions and an example.

\noindent
\textit{Proof of Theorem \ref{thmBEsolution} for the case {\rm \textbf{(\Rmnum{1})}}.}

By \eqref{KPtr}, we have
\begin{eqnarray}
g^{-1} \left(\Lambda_{k}+c \Lambda_{-1} \right)g&=&\sum_{n \in \mathbb{Z}}\left(g^{-1}\psi_{n}g g^{-1} \psi_{n+k}^{*}g+cg^{-1}\psi_{n}g g^{-1}\psi_{n-1}^{*}g\right)\nonumber\\
&=&\sum_{i, j \in \mathbb{Z}}\sum_{n \in \mathbb{Z}}\left(A_{in}(A^{-1})_{n+k,j}+cA_{in}(A^{-1})_{n-1,j}\right)\psi_{i}\psi^{*}_{j}\nonumber,
\end{eqnarray}
which implies
\begin{eqnarray}
\sum_{j \in \mathbb{Z}}\left(A_{mj}(A^{-1})_{j+k,n}+cA_{mj}(A^{-1})_{j-1,n}\right)=f_{m n}=\sum_{i=1}^{N} d_{m}^{(i)} e_{n}^{(i)}.\nonumber
\end{eqnarray}
By \eqref{eH} and \eqref{forim}, the right-hand side of \eqref{ckBE1} is
\begin{eqnarray}\label{rhs}
&\;&\operatorname{Res}_{\lambda}\left((\lambda^{k-1}+c\lambda^{-2}) \langle 1|e^{H(t)}\psi(\lambda)g|0\rangle \langle -1|e^{H(t^{\prime})}\psi^{*}(\lambda)g|0\rangle\right)\nonumber\\
&=&\sum_{m, n \in \mathbb{Z}}\sum_{j \in \mathbb{Z}}\left(A_{mj}(A^{-1})_{j+k,n}+cA_{mj}(A^{-1})_{j-1,n}\right)\langle 1|e^{H(t)}g\psi_{m} |\rangle \langle -1| e^{H(t^{\prime})}g \psi_{n}^{*} |0 \rangle\nonumber\\
&=&\sum_{m, n \in \mathbb{Z}}\sum_{i=1}^{N} d_{m}^{(i)} e_{n}^{(i)} \langle1|e^{H(t)}g\psi_{m} |\rangle \langle -1| e^{H(t^{\prime})}g \psi_{n}^{*} |0 \rangle\nonumber.
\end{eqnarray}
The equation \eqref{rhs} can be rewritten to the left-hand side of \eqref{ckBE1}.

The right-hand side of \eqref{ckBE2} is
\begin{eqnarray}
&\;&\sum_{n \geqslant 0} d_{n}^{(i)}\operatorname{Res}_{\lambda}\left(\lambda^{-1}\left\langle 1\left|e^{H(t)} \psi(\lambda) g\right| 0\right\rangle\left\langle 0\left|e^{H\left(t^{\prime}\right)} \psi^{*}(\lambda) g \psi_{n}\right| 0\right\rangle\right) \nonumber\\
&=&\sum_{l \in \mathbf{Z}}\sum_{n \geqslant 0} d_{n}^{(i)}\left\langle 1\left|e^{H(t)} \psi_{l} g\right| 0\right\rangle\left\langle 0\left|e^{H\left(t^{\prime}\right)} \psi_{l}^{*} g \psi_{n}\right| 0\right\rangle \nonumber\\
&=&\sum_{l, j, s \in \mathbf{Z}}\sum_{n \geqslant 0} d_{n}^{(i)}A_{j, l}(A^{-1})_{l, s}\left\langle 1\left|e^{H(t)} g \psi_{j}\right| 0\right\rangle\left\langle 0\left|e^{H\left(t^{\prime}\right)} g \psi_{s}^{*} \psi_{n}\right| 0\right\rangle\nonumber\\
&=&\sum_{j\geqslant 0}\sum_{n \geqslant 0} d_{n}^{(i)}\left\langle 1\left|e^{H(t)} g \psi_{j}\right| 0\right\rangle\left\langle 0\left|e^{H\left(t^{\prime}\right)} g \psi_{j}^{*} \psi_{n}\right| 0\right\rangle\nonumber\\
&=&\sum_{j \geqslant 0}d_{j}^{(i)}\left\langle 1\left|e^{H(t)} g \psi_{j}\right| 0\right\rangle\left\langle 0\left|e^{H\left(t^{\prime}\right)} g\right| 0\right\rangle.\nonumber
\end{eqnarray}
We can obviously checked that \eqref{ckBE2} is also satisfied. The equation \eqref{ckBE3} can be verified in a similar manner.\hfill\qedsymbol

Next we continue to explore other types of solutions for the $c$-$k$ constrained KP hierarchy.

\begin{proposition}\label{lemN}
Let $g=e^{\Sigma_{m=1}^{l}a_{i_{m}}b_{j_{m}} \psi_{i_{m}}\psi_{j_{m}}^{*}}$ for $i_{m}\geq0, j_{m}<0$, then we have
\begin{eqnarray}
\!\!\!\!\!\!g^{-1}\left(\Lambda_{k}+c\Lambda_{-1}\right)g\!\!\!&=&\!\!\!\Lambda_{k}+\!\!\sum_{m=1}^{l} a_{i_{m}}b_{j_{m}}\left(\psi_{i_{m}-k}\psi_{j_{n}}^{*}-\psi_{i_{m}}\psi_{j_{m}+k}^{*} \right)-\!\!\sum_{n, m=1}^{l} \delta_{i_{m}, j_{n}+k} a_{i_{m}} a_{i_{n}} b_{j_{m}} b_{j_{n}} \psi_{i_{n}}\psi_{j_{m}}^{*}\nonumber\!\!\!\!\!\!\\
&\;&+c\Lambda_{-1}+\sum_{m=1}^{l} ca_{i_{m}}b_{j_{m}}\left(\psi_{i_{m}+1}\psi_{j_{m}}^{*}-\psi_{i_{m}}\psi_{j_{m}-1}^{*} \right)\nonumber.
\end{eqnarray}
\end{proposition}

\begin{proof}
Let $X=\Sigma_{m=1}^{M}a_{i_{m}}b_{j_{m}} \psi_{i_{m}}\psi_{j_{m}}^{*} $. By the Baker-Campbell-Hausdorff formula, we have
\begin{eqnarray}
g^{-1} \Lambda_{k} g= e^{-X} \Lambda_{k} e^{X}=\sum_{n=0}^{\infty} \frac{(-1)^{n}}{n!}(ad X)^{n} (\Lambda_{k}),\nonumber
\end{eqnarray}
where
\begin{eqnarray}
ad X (Y)=[X, Y]. \nonumber
\end{eqnarray}
A straightforward computation shows that
\begin{eqnarray}
ad X(\Lambda_{k})&=&\sum_{m=1}^{M}a_{i_{m}}b_{j_{m}}\left(\psi_{i_{m}}\psi_{j_{m}+k}^{*}-\psi_{i_{m}-k}\psi_{j_{m}}^{*} \right),\nonumber\\
(ad X)^{2}(\Lambda_{k})&=&-\sum_{m,n=1}^{M}2\delta_{i_{m},j_{n}+k}a_{i_{m}}b_{j_{m}}a_{i_{n}}b_{j_{n}}\psi_{i_{n}}\psi_{j_{m}}^{*},\nonumber\\
(ad X)^{n}(\Lambda_{k})&=&0, \;\;n=3, 4, 5,\cdots.\nonumber
\end{eqnarray}
\end{proof}

\begin{proposition}\label{lemNform}
If $g$ is defined by
\begin{eqnarray}\label{gg}
g=e^{\sum_{i,j=1}^{N^{\prime}} a_{i,j}\psi(p_{i})\psi^{*}(q_{j})},  \;\;p_{i}\neq  q_{j},
\end{eqnarray}
then
\begin{eqnarray}
g^{-1}(\Lambda_{k}+c\Lambda_{-1})g=\Lambda_{k}+c \Lambda_{-1}+\sum_{i, j=1}^{N^{\prime}} a_{i, j}\left((p_{i}^{k}-q_{j}^{k})+c(p_{i}^{-1}-q_{j}^{-1})\right)\psi(p_{i})\psi^{*}(q_{j}).\nonumber
\end{eqnarray}
\end{proposition}

\begin{proof}
Let $Y=\sum_{i,j=1}^{N} a_{i,j}\psi(p_{i})\psi^{*}(q_{j})$, similar to the proof of Proposition 3.1, we only need to calculate the following equation.
\begin{eqnarray}
ad Y(\Lambda_{k})=\sum_{i, j=1}^{N^{\prime}} a_{i, j}\left(q_{j}^{k}-p_{i}^{k}\right)\psi(p_{i})\psi^{*}(q_{j}).\nonumber
\end{eqnarray}
By \eqref{phi-com-r}, we can obtain
\begin{eqnarray}
\left[\psi(p_{m})\psi^{*}(q_{n}), \psi(p_{i})\psi^{*}(q_{j})\right]=q_{n}\delta\left(q_{n}-p_{i}\right)\psi(p_{m})\psi^{*}(q_{j})-q_{j}\delta\left(q_{j}-p_{m}\right)\psi(p_{i})\psi^{*}(q_{n}).
\end{eqnarray}
Then a straightforward computation shows that
\begin{eqnarray}
(ad Y)^{2}(\Lambda_{k})=\sum_{m,n,i,j=1}^{N^{\prime}}a_{m,n}a_{i, j}\left(q_{n}\delta\left(q_{n}-p_{i}\right)\psi(p_{m})\psi^{*}(q_{j})-q_{j}\delta\left(q_{j}-p_{m}\right)\psi(p_{i})\psi^{*}(q_{n})\right).\nonumber
\end{eqnarray}
It is easily checked that the last equation vanishes. Thus
\begin{eqnarray}
(ad Y)^{n}(\Lambda_{k})=0, \;\;n=2, 3, 4,\cdots.\nonumber
\end{eqnarray}
\end{proof}

\begin{proposition}\label{corform}
 If $N>1$, $g$ are given by \eqref{gg} and
 \begin{eqnarray}
 a_{i,j} =\sum_{l=1}^{N-1}\frac{d_{i}^{(l)}e_{j}^{(l)}}{(p_{i}^{k}+cp_{i}^{-1})-(q_{j}^{k}+cq_{j}^{-1})},\;\; i, j= 1,2 , \cdots, N^{\prime},
 \end{eqnarray}
 then
 \begin{subequations}
\begin{alignat}{2}
\tau(t)&=\langle 0| e^{H(t)}g|0\rangle,\nonumber\\
\rho_{i}(t)&=\sum_{j=1}^{N} d_{j}^{(i)}\left\langle 1\left|e^{H(t)} g \psi\left(p_{j}\right)\right| 0\right\rangle,\;\;i=1,2, \ldots, N-1 \nonumber\\
\sigma_{i}(t)&=\sum_{j=1}^{N} e_{j}^{(i)}\left\langle-1\left|e^{H(t)} g \psi^{*}\left(q_{j}\right)\right| 0\right\rangle,\;\;i=1,2, \ldots, N-1,\nonumber\\
\rho_{N}(t)&=c\left\langle 1\left|e^{H(t)} g\psi_{0}\right| 0\right\rangle,\nonumber \\
\sigma_{N}(t)&=\left\langle-1\left|e^{H(t)} g\psi^{*}_{-1}\right| 0\right\rangle\nonumber
\end{alignat}
\end{subequations}
satisfy the bilinear equations of the $c$-$k$ constrained KP hierarchy.
In particular, if $g$ is given by \eqref{gg} and $N^{\prime}=N-1$,
 then
 \begin{subequations}
\begin{alignat}{2}
\tau(t)&=\langle 0| e^{H(t)}g|0\rangle,\nonumber\\
\rho_{i}(t)&=\left\langle 1\left|e^{H(t)} g \psi\left(p_{i}\right)\right| 0\right\rangle,\;\;i=1,2, \ldots, N^{\prime},\nonumber\\
\sigma_{i}(t)&=\sum_{j=1}^{N^{\prime}} a_{i j}\left((p_{i}^{k}-q_{j}^{k})+c(p_{i}^{-1}-q_{j}^{-1})\right)\left\langle-1\left|e^{H(t)} g \psi^{*}\left(q_{j}\right)\right| 0\right\rangle,\;\;i=1,2, \ldots, N^{\prime},\nonumber\\
\rho_{N}(t)&=c\left\langle 1\left|e^{H(t)} g\psi_{0}\right| 0\right\rangle,\nonumber \\
\sigma_{N}(t)&=\left\langle-1\left|e^{H(t)} g\psi^{*}_{-1}\right| 0\right\rangle\nonumber
\end{alignat}
\end{subequations}
satisfy the bilinear equations of the $c$-$k$ constrained KP hierarchy.
\end{proposition}

\begin{proof}
By Proposition \ref{lemNform}, we have
\begin{subequations}
\begin{alignat}{2}
g^{-1}(\Lambda_{k}+c\Lambda_{-1})g&=\Lambda_{k}+c \Lambda_{-1}+\sum_{i, j=1}^{N^{\prime}} \sum_{l=1}^{N-1}d_{i}^{(l)}e_{j}^{(l)}\psi(p_{i})\psi^{*}(q_{j})\nonumber\\
&=\Lambda_{k}+c \Lambda_{-1}
+\sum_{m,n}\sum_{l=1}^{N-1}\left(\sum_{i=1}^{N^{\prime}} d_{i}^{(l)}p_{i}^{m}\right) \left(\sum_{j=1}^{N^{\prime}}e_{j}^{(l)}q_{j}^{-n}\right)\psi_{m}\psi^{*}_{n}.\nonumber
\end{alignat}
\end{subequations}
Noting that $\Lambda_{-1}$ contains $\psi_{0} \psi_{-1}^{*}$ , the proposition can be proved.
\end{proof}

Next we give a specific example to show how to obtain rational solutions of the $c$-$k$ constrained KP hierarchy. Noticing that here we give the less trivial solution.

\begin{ex}
Assume $g=e^{a_{i_{1}}b_{j_{1}} \psi_{i_{1}}\psi_{j_{1}}^{*}}$ for $i_{1}\geq0, j_{1}<0$, then from Proposition \ref{lemN} we obtain
\begin{eqnarray}
g^{-1}\left(\Lambda_{k}+c\Lambda_{-1}\right)g&=&\Lambda_{k}+ a_{i_{1}}b_{j_{1}}\left(\psi_{i_{1}-k}\psi_{j_{1}}^{*}-\psi_{i_{1}}\psi_{j_{1}+k}^{*} \right)-\delta_{i_{1}, j_{1}+k} a_{i_{1}}^{2} b_{j_{1}}^{2} \psi_{i_{1}}\psi_{j_{1}}^{*}\nonumber\\
&\;&+c\Lambda_{-1}+ca_{i_{1}}b_{j_{1}}\left(\psi_{i_{1}+1}\psi_{j_{1}}^{*}-\psi_{i_{1}}\psi_{j_{1}-1}^{*} \right)\nonumber.
\end{eqnarray}
We can obtain
\begin{eqnarray}
\tau(t)&=&\langle 0| e^{H(t)}g|0\rangle\nonumber\\
&=&1+a_{i_{1}}b_{j_{1}}\sum_{m=1}^{\infty}p_{m+i_{1}}(t)p_{-j_{1}-m}(t)\nonumber.
\end{eqnarray}
For $N= 1$,  we can obtain the following solution for \eqref{ckKPeq}.

\textup{(1)} When $i_{1}=0, j_{1}=-1$, if $i_{1}=j_{1}+k$,  we cannot obtain a solution of \eqref{ckKPeq} by using Theorem 3.1. If $i_{1}\neq j_{1}+k$,
\begin{eqnarray}
\rho(t)&=&\left\langle 1\left|e^{H(t)} g \left(c\psi_{0}+ca_{i_{1}}\psi_{i_{1}+1}-ca_{i_{1}}\psi_{i_{1}}+a_{i_{1}}\psi_{i_{1}-k}-a_{i_{1}}\psi_{i_{1}}\right)\right| 0\right\rangle, \nonumber\\
&=&ca_{0}X_{1}\tau(t)+c-ca_{0}-a_{0}\nonumber,
\end{eqnarray}
\begin{eqnarray}
\sigma(t)&=&\left\langle-1\left|e^{H(t)} g \left( \psi_{-1}^{*}+b_{j_{1}}\psi_{j_{1}}^{*}+b_{j_{1}}\psi_{j_{1}-1}^{*}+b_{j_{1}}\psi_{j_{1}}^{*}+b_{j_{1}}\psi_{j_{1}+k}^{*}\right)\right| 0\right\rangle, \nonumber\\
&=&b_{-1}X_{-1}^{*}\tau(t)+2b_{-1}+1.\nonumber
\end{eqnarray}

\textup{(2)} When $i_{1}\geq1,  j_{1}<-1$, if $i_{1}=j_{1}+k$,  we cannot obtain a solution of \eqref{ckKPeq} by using Theorem 3.1. If $i_{1}\neq j_{1}+k$,
\begin{eqnarray}
\rho(t)&=&\left\langle 1\left|e^{H(t)} g \left(c\psi_{0}+ca_{i_{1}}\psi_{i_{1}+1}-ca_{i_{1}}\psi_{i_{1}}+a_{i_{1}}\psi_{i_{1}-k}-a_{i_{1}}\theta_{i_{1},j_{1}+k}\psi_{i_{1}}-\delta_{i_{1}, j_{1}+k}a_{i_{1}}^{2}\psi_{i_{1}}\right)\right| 0\right\rangle, \nonumber\\
&=&cX_{0}\tau(t)+ca_{i_{1}}X_{i_{1}+1}\tau(t)-ca_{i_{1}}p_{i_{1}}(t)+a_{i_{1}}\theta_{i_{1},j_{1}+k}\psi_{i_{1}}X_{i_{1}-k}\tau(t)-\delta_{i_{1}, j_{1}+k}a_{i_{1}}^{2}p_{i_{1}}(t),\nonumber
\end{eqnarray}
\begin{eqnarray}
\sigma(t)&=&\left\langle-1\left|e^{H(t)} g \left( \psi_{-1}^{*}+b_{j_{1}}\psi_{j_{1}}^{*}+b_{j_{1}}\psi_{j_{1}-1}^{*}+b_{j_{1}}\psi_{j_{1}}^{*}+b_{j_{1}}\theta_{i_{1},j_{1}+k}\psi_{j_{1}+k}^{*}+\delta_{i_{1}, j_{1}+k}b_{j_{1}}^{2} \psi_{j_{1}}^{*} \right)\right| 0\right\rangle,\nonumber\\
&=&X_{0}^{*}\tau(t)+2b_{j_{1}}p_{j_{1}}(-t)+b_{j_{1}}X_{j_{1}}^{*}+b_{j_{1}}\theta_{i_{1},j_{1}+k}X_{j_{1}+k+1}^{*}\tau(t)+\delta_{i_{1}, j_{1}+k}b_{j_{1}}^{2}p_{j_{1}}(-t),\nonumber
\end{eqnarray}
where $\theta_{i_{1},j_{1}+k}$ is defined by
\begin{eqnarray}
\theta_{i_{1},j_{1}+k}= \begin{cases}1, & \text { for } \quad i_{1}\neq j_{1}+k \\ 0, & \text { for } i_{1}= j_{1}+k\end{cases}.\nonumber
\end{eqnarray}

For $N=2$. The solution is
\begin{eqnarray}
\rho_{1}(t)&=&\left\langle 1\left|e^{H(t)} g \left(c\psi_{0}+ca_{i_{1}}\psi_{i_{1}+1}-ca_{i_{1}}\psi_{i_{1}}\right)\right| 0\right\rangle, \nonumber\\
&=&c+ca_{i_{1}}X_{i_{1}+1}\tau(t)-ca_{i_{1}}p_{i_{1}}(t),\nonumber
\end{eqnarray}
\begin{eqnarray}
\sigma_{1}(t)&=&\left\langle-1\left|e^{H(t)} g \left( \psi_{-1}^{*}+b_{j_{1}}\psi_{j_{1}}^{*}+b_{j_{1}}\psi_{j_{1}-1}^{*}\right)\right| 0\right\rangle,\nonumber\\
&=&1+b_{j_{1}}p_{i_{1}}(-t)+b_{j_{1}X_{j_{1}}}X^{*}\tau(t),\nonumber
\end{eqnarray}
\begin{eqnarray}
\rho_{2}(t)&=&\left\langle 1\left|e^{H(t)} g \left(a_{i_{1}}\psi_{i_{1}-k}-a_{i_{1}}\theta_{i_{1},j_{1}+k}\psi_{i_{1}}-\delta_{i_{1}, j_{1}+k}a_{i_{1}}^{2}\psi_{i_{1}}\right)\right| 0\right\rangle, \nonumber\\
&=&a_{i_{1}}X_{i_{1}-k}\tau(t)-a_{i_{1}}\theta_{i_{1},j_{1}+k}p_{i_{1}}(t)-\delta_{i_{1}, j_{1}+k}a_{i_{1}}^{2}p_{i_{1}}(t),\nonumber
\end{eqnarray}
\begin{eqnarray}
\sigma_{2}(t)&=&\left\langle-1\left|e^{H(t)} g \left(b_{j_{1}}\psi_{j_{1}}^{*}+b_{j_{1}}\theta_{i_{1},j_{1}+k}\psi_{j_{1}+k}^{*}+\delta_{i_{1}, j_{1}+k}b_{j_{1}}^{2}\psi_{j_{1}}^{*}\right)\right| 0\right\rangle,\nonumber\\
&=&b_{j_{1}}p_{j_{1}}(-t)+b_{j_{1}}\theta_{i_{1},j_{1}+k}X_{j_{1}+k+1}\tau(t)+\delta_{i_{1}, j_{1}+k}b_{j_{1}}^{2}p_{j_{1}}(-t).\nonumber
\end{eqnarray}

\end{ex}

\subsection{ The case of the $c$-$k$ constrained BKP hierarchy}

We here prove Theorem \ref{thmBEsolution} for the case {\rm \textbf{(\Rmnum{2})}} and give several conclusions and an example.

\noindent
\textit{Proof of Theorem \ref{thmBEsolution} for the case {\rm \textbf{(\Rmnum{2})}}.}

By \eqref{BKPusefor}, the right hand side of the \eqref{cBKPfuntau1} is
\begin{eqnarray}
&\;&\operatorname{Res}_\lambda\left(\lambda^{k-1}+c \lambda^{-2}\right) \tilde{\tau}\left(\tilde{t}-2\tilde{\lambda}\right) \tilde{\tau}\left(\tilde{t}^{\prime}+2\tilde{\lambda}\right) e^{\tilde{\xi}\left(\tilde{t}-\tilde{t}^{\prime}, \lambda\right)}\nonumber\\
&=&\operatorname{Res}_\lambda\left(\lambda^{k-1} X_B(\tilde{t}, \lambda) \tilde{\tau}(\tilde{t}) X_B\left(\tilde{t}^{\prime},-\lambda\right) \tilde{\tau}\left(\tilde{t}^{\prime}\right)+c \lambda^{-2} X_B(\tilde{t}, \lambda) \tilde{\tau}(\tilde{t}) X_B\left(\tilde{t}^{\prime},-\lambda\right) \tilde{\tau}\left(\tilde{t}^{\prime}\right)\right)\nonumber\\
&=&4 \operatorname{Res}_\lambda \lambda^{k-1}\langle 0\left|\phi_0 e^{H_{B}(\tilde{t})} \phi(\lambda) \tilde{g}\right| 0 \rangle \langle 0\left|\phi_0 e^{H_{B}(\tilde{t})} \phi(-\lambda) \tilde{g}\right| 0 \rangle\nonumber\\
&\;&+4c \operatorname{Res}_{\lambda} \lambda^{-2}\langle 0\left|\phi_0 e^{H_{B}(\tilde{t})} \phi(\lambda) \tilde{g}\right| 0 \rangle \langle 0\left|\phi_0 e^{H_{B}(\tilde{t})} \phi(-\lambda) \tilde{g}\right| 0 \rangle\nonumber\\
&=&4 \sum_i\left\langle 0\left|\phi_0 e^{H_B(\tilde{t})} \phi_i \tilde{g}\right| 0\right\rangle\left\langle 0\left|\phi_0 e^{H_B\left(\tilde{t}^{\prime}\right)} \phi_{i+k}^* \tilde{g}\right| 0\right\rangle \nonumber\\
&\;&+4 c \sum_j\left\langle 0\left|\phi_0 e^{H_B(\tilde{t})} \phi_j \tilde{g}\right| 0\right\rangle\left\langle 0\left|\phi_0 e^{H_B\left(\tilde{t}^{\prime}\right)} \phi_{j=1}^* \tilde{g}\right| 0\right\rangle \nonumber\\
&=& 4 \sum_i\left\langle 0\left|\phi_0 e^{H_B(\tilde{t})} \tilde{g} \sum_l \tilde{a}_{l i} \phi_l\right| 0\right\rangle\left\langle 0\left|\phi_0 e^{H_B\left(\tilde{t}^{\prime}\right)} \tilde{g} \sum_m \tilde{b}_{i+k, m} \phi_m^*\right| 0\right\rangle \nonumber\\
&\;&+4 c \sum_j\left\langle 0\left|\phi_0 e^{H_{B}(\tilde{t})}\tilde{g} \sum_s \tilde{a}_{s j} \phi_l\right| 0\right\rangle\left\langle 0\left|\phi_0 e^{H_B\left(\tilde{t}^*\right)} \tilde{g} \sum_n \tilde{b}_{j-1, n} \phi_m^*\right| 0\right\rangle. \nonumber
\end{eqnarray}
A direct calculation shows
\begin{eqnarray}
\tilde{g}^{-1}\left(\tilde{\Lambda}_k+c \tilde{\Lambda}_{-1}\right) \tilde{g} &=&\sum_n\left(\tilde{g}^{-1} \phi_n \tilde{g}\right)\left(\tilde{g}^{-1} \phi_{n+k}^* \tilde{g}\right)+c \sum_m\left(\tilde{g}^{-1} \phi_m \tilde{g}\right)\left(\tilde{g}^{-1} \phi_{m-1}^* \tilde{g}\right) \nonumber\\
&=&\sum_{l, j}\left(\sum_n \tilde{a}_{l n} \tilde{b}_{n+k, j} \phi_i \phi_j^*\right)+c \sum_{l, j}\left(\sum_m \tilde{a}_{l m} \tilde{b}_{m-1, j} \phi_l \phi_j^*\right).\nonumber
\end{eqnarray}
Then
\begin{eqnarray}
\sum_{h, j}\left(\sum_n \tilde{a}_{l n} \tilde{b}_{n+k, j}+c \sum_m \tilde{a}_{l m} \tilde{b}_{m-1, j}\right)=\tilde{f}_{i j}=\sum_{i=1}^{\tilde{N}}\left(\tilde{d}_l^{(i)} \tilde{e}_j^{(i)}-(-1)^{l+j} \tilde{e}_{-l}^{(i)} \tilde{d}_{-j}^{(i)}\right).\nonumber
\end{eqnarray}
The equation \eqref{cBKPeigenfun1} is true through the following expression
\begin{eqnarray}
&\;&4 \sum_{i=1}^{\tilde{N}}\left\langle 0\left|\phi_0 e^{H_B(\tilde{t})} \tilde{g} \sum_{l \geq 0} \tilde{d}_l^{(i)} \phi_l\right| 0\right\rangle\left\langle 0\left|\phi_0 e^{H_{B}\left(\tilde{t}^{\prime}\right)} \tilde{g} \sum_{j \leq 0} \tilde{e}_j^{(i)} \phi_j^*\right| 0\right\rangle \nonumber\\
&\;&-4 \sum_{i=1}^{\tilde{N}}\left\langle 0\left|\phi_0 e^{H_B(\tilde{t})} g \sum_{l \geq 0}(-1)^l \tilde{e}_{-l}^{(i)} \phi_l\right| 0\right\rangle\left\langle 0\left|\phi_0 e^{H_B\left(\tilde{t}^{\prime}\right)} \tilde{g} \sum_{j \leq 0}(-1)^j \tilde{d}_{-j}^{(i)} \phi_j^*\right| 0\right\rangle \nonumber\\
&=&\sum_{i=1}^{\tilde{N}}\left(\tilde{\sigma}_i(\tilde{t}) \tilde{\rho}_i\left(\tilde{t}^{\prime}\right)-\tilde{\rho}_i(\tilde{t}) \tilde{\sigma}_i\left(\tilde{t}^{\prime}\right)\right)\nonumber
\end{eqnarray}
Next, we obtain
\begin{eqnarray}
&\;&\operatorname{Res}_{\lambda} \left(\lambda^{-1} X_B(\tilde{t}, \lambda) \tilde{\tau}(\tilde{t}) X_B\left(\tilde{t}^{\prime},-\lambda\right) \tilde{\sigma}_i\left(\tilde{t}^{\prime}\right)\right) \nonumber\\
&=&4 \operatorname{Res}_\lambda \lambda^{-1}\left\langle 0 \phi_{0}e^{H_B(\tilde{t})} \phi_{(\lambda)} \tilde{g} \mid 0\right\rangle\left\langle 0\left|e^{H_B\left(\tilde{t}^{\prime}\right)} \phi_{(-\lambda)} \tilde{g} \sum_{l \geq 0} \tilde{d}_l^{i} \phi_l\right| 0\right\rangle\nonumber \\
&=& 4 \sum_{m, i, j \in Z}\left\langle 0\left|\phi_0 e^{H_B(\tilde{t})} \tilde{g} \phi_i^*\right| 0\right\rangle \tilde{a}_{m i} \tilde{b}_{j m}\left\langle 0\left|e^{H_B\left(\tilde{t}^{\prime}\right)} \tilde{g} \phi_j \sum_{l \geq 0} \tilde{d}_l^{(i)} \phi_l\right| 0\right\rangle \nonumber\\
&=&4 \sum_{m \leq 0}\left\langle 0\left|\phi_0 e^{H_B\left(\tilde{t}^{\prime}\right)} \tilde{g} \phi_m^*\right| 0\right\rangle\left\langle 0\left|e^{H_B\left(\tilde{t}^{\prime}\right)} \tilde{g} \phi_m \sum_{l \geq 0} \tilde{d}_l^{(i)} \phi_l\right| 0\right\rangle \nonumber\\
&=&4 \sum_{m \leq 0}\left\langle 0 \mid \phi_0 e^{H_B\left(\tilde{t}^{\prime}\right)} \tilde{g} \phi_m^*\right\rangle\left\langle 0\left|e^{H_B\left(\tilde{t}^{\prime}\right)} \tilde{g} \sum_{l \geq 0} \tilde{d}_l^{(i)}\left((-1)^{\prime} \delta_{l+j, 0}-\phi_i \phi_m\right)\right| 0\right\rangle \nonumber\\
&=&2 \tilde{\sigma}_i(\tilde{t}) \tilde{\tau}\left(\tilde{t}^{\prime}\right)-\tilde{\tau}\left(\tilde{t}^{\prime}\right) \tilde{\sigma}_i\left(\tilde{t}^{\prime}\right).\nonumber
\end{eqnarray}
Then the \eqref{cBKPeigenfun2} is true, and the proof of \eqref{cBKPeigenfun3} is similar.\hfill\qedsymbol

Next, we specifically discuss a case. Setting $Y=\sum_{n=1}^{\widehat{N}} a_n \phi_{i_n} \phi_{j_n}^*, \tilde{g}=e^Y$, we can explicitly express the results.
\begin{proposition}\label{cBKPpropY}
Assuming $Y=\sum_{n=1}^{\tilde{N}} a_n \phi_{i_n} \phi_{j_n}^*, \tilde{g}=e^Y$, with $i_1>-j_1>0$ and $i_n>-j_n \geq$ $0(n \geq 2)$, then $\tilde{g}^{-1}\left(\tilde{\Lambda}_k+c \tilde{\Lambda}_{-1}\right) \tilde{g}$ can be expressed as
\begin{eqnarray}
&\tilde{\Lambda}_k+c \tilde{\Lambda}_{-1}-2 \sum_{n=1}^{\tilde{N}} a_n \left(\phi_{i_n} \phi_{j_n+k}^*-\phi_{i_n-k} \phi_{j_n}^*\right)-2 c \sum_{n=1}^{\tilde{N}} a_n\left(\phi_{i_n} \phi_{j_n-1}^*-\phi_{i_n+1} \phi_{j_n}^*\right) \nonumber\\
&+\sum_{l, n=1}^{\tilde{N}} a_l a_n(-1)^{j_l+1} \delta_{j_l+j_n+k, 0} \phi_{i_l} \phi_{i_n}+(-1)^{i_l+1} \delta_{i_l+i_n-k, 0} \phi_{j_n}^* \phi_{j_l}^*-2 \delta_{i_l, j_n+k} \phi_{i_n} \phi_{j_l}^* \nonumber\\
&+c \sum_{l, n=1}^{\tilde{N}} a_l a_n(-1)^{j_l+1} \delta_{j_l+j_n-1,0} \phi_{i_l} \phi_{i_n}+(-1)^{i_l+1} \delta_{i_l+i_n+1,0} \phi_{j_n}^* \phi_{j_l}^*-2 \delta_{i_l, j_n-1} \phi_{i_n} \phi_{j_l}^* \nonumber\\
&+a_1^2 \delta_{j_1, 0} \sum_{n=1}^{\tilde{N}} a_n\left((-1)^{i_n} \delta_{i_n+i_1-k, 0} \phi_{i_1} \phi_{j_n}^*+\delta_{j_n, i_1-k} \phi_{i_n} \phi_{i_1}\right)+a_1^2 \delta_{j_1, 0} \phi_{i_1} \phi_{i_1-k} \nonumber\\
&+c a_1^2 \delta_{j_1, 0} \sum_{n=1}^{\tilde{N}} a_n\left((-1)^{i_n} \delta_{i_n+i_1+1,0} \phi_{i_1} \phi_{j_n}^*+\delta_{j_n, i_1+1} \phi_{i_n} \phi_{i_1}\right)+c a_1^2 \delta_{j_1, 0} \phi_{i_1} \phi_{i_1+1}.\nonumber
\end{eqnarray}
\end{proposition}

\begin{proof}
By the Baker-Campbell-Hausdorff formula, we have
\begin{eqnarray}
\tilde{g}^{-1} \tilde{\Lambda}_k \tilde{g}=e^{-Y} \tilde{\Lambda}_k e^Y=\sum_{n=0} \frac{(-1)^n}{n!}(a d Y)^n\left(\tilde{\Lambda}_k\right),\nonumber
\end{eqnarray}
where $(ad Y) \tilde{\Lambda}_k=\left[Y, \tilde{\Lambda}_k\right]$. Obtain that
\begin{eqnarray}
\left[Y, \tilde{\Lambda}_k+c \Lambda_{-1}\right]&=&2 \sum_{n=1}^{\tilde{N}} a_n\left(\phi_{i_n} \phi_{j_n+k}^*-\phi_{i_n-k} \phi_{j_n}^*+c \phi_{i_n} \phi_{j_n-1}^*-c \phi_{i_n+1} \phi_{j_n}^*\right),\nonumber \\
(a d Y)^2\left(\Lambda_k+c \Lambda_{-1}\right)&=&2 \sum_{l, n=1}^{\tilde{N}} a_l a_n\left((-1)^{j_l+1} \delta_{j_l+j_n+k, 0} \phi_{i_l} \phi_{i_n}+(-1)^{i_l+1} \delta_{i_l+i_n-k, 0} \phi_{j_n}^* \phi_{j_l}^*\right.\nonumber\\
&\;&-2 \delta_{i_l, j_n+k} \phi_{i_n} \phi_{j_l}^*+(-1)^{j_l+1} c \delta_{j_l+j_n-1,0} \phi_{i_l} \phi_{i_n} \nonumber\\
&\;&\left.+(-1)^{i_l+1} c \delta_{i_l+i_n+1,0} \phi_{j_n}^* \phi_{j_l}^*-2 c \delta_{i_l, j_n-1} \phi_{i_n} \phi_{j_l}^*\right),\nonumber \\
(a d Y)^3\left(\tilde{\Lambda}_k+c \tilde{\Lambda}_{-1}\right)&=&-6 a_1^2 \delta_{j_1, 0} \sum_{n=1}^{\tilde{N}} a_n\left((-1)^{i_n} \delta_{i_n+i_1-k, 0} \phi_{i_1} \phi_{j_n}^*+\delta_{j_n, i_l-k} \phi_{i_n} \phi_{i_1}\right.\nonumber\\
&\;&\left.+(-1)^{i_n} c \delta_{i_n+i_1+1,0} \phi_{i_1} \phi_{j_n}^*+c \delta_{j_n, i_1+1} \phi_{i_n} \phi_{i_1}\right), \nonumber\\
(a d Y)^n\left(\tilde{\Lambda}_k+c \tilde{\Lambda}_{-1}\right)&=& 0, \quad n \geq 4.\nonumber
\end{eqnarray}
Let $k=-1$, and a direct computation is shown in the proposition.
\end{proof}

Next, we will consider other solutions of the $c$-$k$ constrained BKP hierarchy. Give $\tilde{g}$ the expression
\begin{eqnarray}
\tilde{g}=e^{\sum_{i, j=1}^N a_{i j} \phi\left(p_i\right) \phi\left(-q_j\right)} \equiv e^Y,\nonumber
\end{eqnarray}
with $p_i \neq \pm q_j, p_i \neq-p_j$ and $q_i \neq-q_j$. Then we have the following property.

\begin{proposition}
Assuming $Y=\sum_{i, j=1}^{\tilde{N}} a_{i j} \phi\left(p_i\right) \phi\left(-q_j\right), \tilde{g}=e^Y$, we have
\begin{eqnarray}
\tilde{g}^{-1}\left(\tilde{\Lambda}_k+c \tilde{\Lambda}_{-1}\right) g=\tilde{\Lambda}_k+c \tilde{\Lambda}_{-1}+2 \sum_{i, j=1}^{\tilde{N}} a_{i j}\left(p_i{ }^k-q_j{ }^k+p_i{ }^{-1}-q_j{ }^{-1}\right) \phi\left(p_j\right) \phi\left(-q_j\right).\nonumber
\end{eqnarray}
\end{proposition}

\begin{proof}
By straightforward computation, there are the following equations,
\begin{eqnarray}
\left[Y, \tilde{\Lambda}_k\right]=-2 \sum_{i, j=1}^{\tilde{N}} a_{i j}\left(p_i{ }^k-q_j{ }^k\right) \phi\left(p_i\right) \phi\left(-q_j\right),\nonumber
\end{eqnarray}
and
\begin{eqnarray}
\left[Y,\left[Y, \tilde{\Lambda}_k\right]\right]&=&-2 \sum_{m, n=1}^{\tilde{N}} \sum_{i, j=1}^{\tilde{N}} a_{m n} a_{i j}\left(p_i{ }^k-q_j{ }^k\right)\left[\phi\left(p_m\right) \phi\left(-q_s\right), \phi\left(p_i\right) \phi\left(-q_j\right)\right]\nonumber \\
&=&0.\nonumber
\end{eqnarray}
Setting $k=-1$, we can obtain the commutation relation of $Y$ and $c \tilde{\Lambda}_{-1}$.
\end{proof}

\begin{proposition}
If $\tilde{g}$ satisfies the condition
\begin{eqnarray}
\tilde{g}^{-1}\left(\tilde{\Lambda}_k+c \tilde{\Lambda}_{-1}\right) \tilde{g}=\tilde{\Lambda}_k+c \tilde{\Lambda}_{-1}+\sum_{l, j} f_{l, j} \phi_l \phi_s^*,\nonumber
\end{eqnarray}
with
\begin{eqnarray}\label{condi}
\tilde{f}_{l j}=\sum_{n=1}^m\left(\tilde{d}_r^{(n)} \tilde{e}_s^{(n)}-(-1)^{l+j} \tilde{d}_{-s}^{(n)} \tilde{e}_{-r}^{(n)}\right),
\end{eqnarray}
where $\tilde{d}_r^{(n)}=\sum_{i=1}^N d_i^{(n)} p_i^l$ and $\tilde{e}_j^{(n)}=\sum_{m=1}^N e_m^{(n)} q_m^{-j}$, then the expression of $\tilde{\tau}(\tilde{t})$ and
\begin{eqnarray}
\tilde{\sigma}_i(\tilde{t}) &=&2\sum_{r=1}^N d_r^{(i)} \left\langle 0\left|\phi_0 e^{H_B(\tilde{t})} \tilde{g} \phi\left(p_r\right)\right| 0\right\rangle, \nonumber\\
\tilde{\rho}_i(\tilde{t}) &=&2 \sum_{r=1}^{\tilde{N}} e_r^{(i)}\left\langle 0\left|\phi_0 e^{H_B(\tilde{t})} \tilde{g} \phi\left(-q_r\right)\right| 0\right\rangle, \quad i=1,2,3, \ldots, \tilde{N}-1, \nonumber\\
\tilde{\sigma}_{\tilde{N}}(t) &=&2 c\left\langle 0\left|\phi_0 e^{H_B(\tilde{t})} \tilde{g} \phi_0\right| 0\right\rangle, \nonumber\\
\tilde{\rho}_{\tilde{N}}(\tilde{t}) &=&2\left\langle 0\left|\phi_0 e^{H_B(\tilde{t})} \tilde{g} \phi_{-1}^*\right| 0\right\rangle\nonumber
\end{eqnarray}
solve the $c$-$k$ constrained BKP hierarchy.

The condition \eqref{condi} refers to $a_{i j}\left(\left(p_i{ }^k+c p_i{ }^{-1}\right)-\left(q_j{ }^k+c q_j{ }^{-1}\right)\right)=\sum_{l=1}^{\tilde{N}-1} d_i^{(l)} e_j^{(l)}$, then the solutions of the $c$-$k$ constrained BKP hierarchy can be simplified as
\begin{eqnarray}
\tilde{\tau}(\tilde{t})&=&\left\langle 0\left|e^{H_B(\tilde{t})} \tilde{g}\right| 0\right\rangle,\nonumber\\
\tilde{\sigma}_i(\tilde{t})&=&2\left\langle 0\left|\phi_0 e^{H_B(\tilde{t})} \tilde{g} \phi\left(p_i\right)\right| 0\right\rangle,\nonumber\\
\tilde{\rho}_i(\tilde{t})&=&2 \sum_{j=1}^N a_{i j}\left(\left(p_i^k+c p_i^{-1}\right)-\left(q_j^k+c q_j{ }^{-1}\right)\right)\left\langle 0\left|\phi_0 e^{H_B(\tilde{t})} g \phi\left(-q_j\right)\right| 0\right\rangle.\nonumber
\end{eqnarray}
\end{proposition}

Next, we give an example.
\begin{ex}
For $\tilde{N}=1$, let $g=e^{a_1 \phi_{i_1} \phi_{j_1}^*}=1+a_1 \phi_{i_1} \phi_{j_1}^*$, when $i_1>-j_1 \geq 0$. The tau function is
\begin{eqnarray}
\tilde{\tau}(\tilde{t})=\langle 0| e^{H_B(\tilde{t})}\left(1+a_1 \phi_{i_1} \phi_{j_1}^*\right)|0\rangle.\nonumber
\end{eqnarray}
To compute the auxiliary functions $\tilde{\sigma}_i(\tilde{t})$ and $\tilde{\rho}_i(\tilde{t})$, we need the following result
\begin{eqnarray}
\tilde{g}^{-1}\left(\tilde{\Lambda}_k+c \tilde{\Lambda}_{-1}\right) \tilde{g} &=&\tilde{\Lambda}_k+c \tilde{\Lambda}_{-1}-2 a_1\left(\phi_{i_1} \phi_{j_1+k}^*-\phi_{i_1-k} \phi_{j_1}^*+c \phi_{i_1} \phi_{j_1-1}^*-c \phi_{i_1+1} \phi_{j_1}^*\right) \nonumber\\
&\;&-2 a_1^2\left(\delta_{i_1, j_1+k} \phi_{i_1} \phi_{j_1}^*+c \delta_{i_1, j_1-1} \phi_{i_1} \phi_{j_1}^*\right).\nonumber
\end{eqnarray}
Next, we discuss the simplest case of $\tilde{N}=1$, and divide it into two cases of $j_1<0$ and $j_1=0$. For $j_1<0$, we have the following four cases.

(1) If $k>i_1$ and $k+j_1=i_1$, the solution is
\begin{eqnarray}
\tilde{\sigma}_1(\tilde{t})&=&2\left\langle 0\left|\phi_0 e^{H_B(\tilde{t})}\left(1+a_1 \phi_{i_1} \phi_{j_1}^*\right)\left(a_1 \phi_{i_1}-c \phi_{i_1}+a_1 \phi_{i_1+1}+c \phi_0\right)\right| 0\right\rangle, \nonumber\\
\tilde{\rho}_1(\tilde{t})&=&2\left\langle 0\left|\phi_0 e^{H_B(\tilde{t})}\left(1+a_1 \phi_{i_1} \phi_{j_1}^*\right)\left(-a_1 \phi_{j_1}^*+a_1 \phi_{j_1-1}^*+c \phi_{j_1}^*+\phi_{-1}^*\right)\right| 0\right\rangle.\nonumber
\end{eqnarray}

(2) If $k>i_1$ and $k+j_1 \neq i_1$, the solution is
\begin{eqnarray}
\tilde{\sigma}_1(\tilde{t})&=&2\left\langle 0\left|\phi_0 e^{H_B(\tilde{t})}\left(1+a_1 \phi_{i_1} \phi_{j_1}^*\right)\left(-a_1 \phi_{i_1}+a_1 \phi_{i_1+1}+c \phi_0\right)\right| 0\right\rangle, \nonumber\\
\tilde{\rho}_1(\tilde{t})&=&2\left\langle 0\left|\phi_0 e^{H_B(\tilde{t})}\left(1+a_1 \phi_{i_1} \phi_{j_1}^*\right)\left(c \phi_{j_1-1}^*+c \phi_{j_1}^*+\phi_{-1}^*\right)\right| 0\right\rangle.\nonumber
\end{eqnarray}

(3) If $-j_1<k \leq i_1$, the solution is
\begin{eqnarray}
\tilde{\sigma}_1(\tilde{t})&=&2\left\langle 0\left|\phi_0 e^{H_B(\tilde{t})}\left(1+a_1 \phi_{i_1} \phi_{j_1}^*\right)\left(\phi_{i_1-k}+c \phi_0-a_1 \phi_{i_1}+a_1 \phi_{i_1+1}\right)\right| 0\right\rangle, \nonumber\\
\tilde{\rho}_1(\tilde{t})&=&2\left\langle 0\left|\phi_0 e^{H_B(\tilde{t})}\left(1+a_1 \phi_{i_1} \phi_{j_1}^*\right)\left(a_1 \phi_{j_1}^*+\phi_{-1}^*+c \phi_{j_1-1}^*+c \phi_{j_1}^*\right)\right| 0\right\rangle.\nonumber
\end{eqnarray}

(4) If $0<i_1<k \geq-j_1, i_1 \neq j_1+k$, the solution is
\begin{eqnarray}
\tilde{\sigma}_1(\tilde{t})&=&2\left\langle 0\left|\phi_0 e^{H_B(\tilde{t})}\left(1+a_1 \phi_{i_1} \phi_{j_1}^*\right)\left(\phi_{i_1}+c \phi_0-a_1 \phi_{i_1}+a_1 \phi_{i_1+1}\right)\right| 0\right\rangle, \nonumber\\
\tilde{\rho}_1(\tilde{t})&=&2\left\langle 0\left|\phi_0 e^{H_B(\tilde{t})}\left(1+a_1 \phi_{i_1} \phi_{j_1}^*\right)\left(-a_1 \phi_{j_1+k}^*+\phi_{-1}^*+c \phi_{j_1-1}^*+c \phi_{j_1}^*\right)\right| 0\right\rangle.\nonumber
\end{eqnarray}

For $j_1=0$, the tau function is
\begin{eqnarray}
\tilde{\tau}(\tilde{t})=\left\langle 0\left|e^{H_B(\tilde{t})}\left(1+a_1 \phi_{i_1} \phi_0^*\right)\right| 0\right\rangle.\nonumber
\end{eqnarray}
There are three cases for the the auxiliary functions $\sigma_i(t)$ and $\rho_i(t)$.

(1) If $i_1>k>0$, the solution is
\begin{eqnarray}
&\tilde{\sigma}_1(\tilde{t})=2\left\langle 0\left|\phi_0 e^{H_B(\tilde{t})}\left(1+a_1 \phi_{i_1} \phi_0^*\right)\left(\phi_{i_1-k}+c \phi_0-c \phi_{i_1}+a_1 \phi_{i_1+1}\right)\right| 0\right\rangle, \nonumber\\
&\tilde{\rho}_1(\tilde{t})=2\left\langle 0\left|\phi_0 e^{H_B(\tilde{t})}\left(1+a_1 \phi_{i_1} \phi_0^*\right)\left(\left(a_1 \phi_0^*+\frac{a_1^2}{2} \phi_{i_1}\right)+\phi_{-1}^*+\left(a \phi_{-1}^*-\frac{a_1^2}{2} \phi_{i_1+1}\right)+c \phi_0^*\right)\right| 0\right\rangle.\nonumber
\end{eqnarray}
(2) If $i_1<k$, the solution is
\begin{eqnarray}
\tilde{\sigma}_1(\tilde{t})&=&2\left\langle 0\left|\phi_0 e^{H_B(\tilde{t})}\left(1+a_1 \phi_{i_1} \phi_0^*\right)\left(c \phi_0-a_1\left(\phi_{i_1}-\phi_{i_1+1}\right)+\frac{1}{2} a_1^2 \phi_{i_1}\right)\right| 0\right\rangle, \nonumber\\
\tilde{\rho}_1(\tilde{t})&=&2\left\langle 0\left|\phi_0 e^{H_B(\tilde{t})}\left(1+a_1 \phi_{i_1} \phi_0^*\right)\left(\phi_{-1}^*+c\left(\phi_{j_1-1}^*+\phi_{j_1}^*\right)+\phi_{i_1+1}\right)\right| 0\right\rangle.\nonumber
\end{eqnarray}
\end{ex}

\section{Additional symmetries}

In this section, we study the additional symmetry of the $c$-$k$ constrained KP and BKP hierarchies and its action on eigenfunction $q_{i}(t)$ and adjoint eigenfunction $r_{i}(t)$. For convenience, we take $k=1$ in this paper, and for $k>1$ we can discuss it in a similar way.

For the KP hierarchy, we recall the dressing operator
\begin{eqnarray}
\phi=1+w_{i}\partial^{-1}+w_{2}\partial^{-2}+\cdots, \nonumber
\end{eqnarray}
which satisfies $L=\phi \partial \phi^{-1}$
and the Orlov-Shulman operator \cite{OrlovA} $M=\phi \Gamma \phi^{-1}$, in which $\Gamma=\sum_{i=1}^{\infty}it_{i}\partial ^{i-1}$.
For the BKP hierarchy, we review the dressing operator
\begin{eqnarray}
\tilde{\phi}=1+\tilde{w}_{i}\partial^{-1}+\tilde{w}_{2}\partial^{-2}+\cdots, \nonumber
\end{eqnarray}
which satisfies $\tilde{L}=\tilde{\phi }\partial \tilde{\phi}^{-1}$
and the Orlov-Shulman operator \cite{TuOtB} $\tilde{M}=\tilde{\phi }\tilde{\Gamma} \tilde{\phi}^{-1}$, in which $\tilde{\Gamma}=\sum_{i=0}^{\infty}(2i+1)t_{2i+1}\partial ^{2i}$.

We list the main result of this section in the following theorem. It is clear that the additional symmetry of the $c$-$1$ constrained BKP hierarchies form the subalgebra of the additional symmetry of the constrained KP hierarchies.
\begin{thm}\label{thmadditonalsym}
The additional symmetry of the $c$-$1$ constrained KP and BKP hierarchies can be given as follows.

{\rm \textbf{(\Rmnum{1}).}} For the $c$-$1$ constrained KP hierarchy, its additional symmetry is given by
\begin{eqnarray}\label{addLax}
\frac{\partial L}{\partial t_{n}^{*}}=[-\left(\Pi_{n}\right)_{-}-X_{n+2}+4cX_{n}, L ],
\end{eqnarray}
where
\begin{subequations}\label{X}
\begin{alignat}{2}
X_{n}&=0,\;\;n=0, 1, 2,\\
X_{n}&=\sum_{i=1}^{N}\sum_{j=0}^{n-2}\left(j-\frac{1}{2}(n-2)\right)(L+cL^{-1})^{n-2-j}(q_{i})\partial^{-1}\circ (L+cL^{-1})^{*j}(r_{i}),\;\;n=3,4,5,\cdots,
\end{alignat}
\end{subequations}
and
\begin{eqnarray}
\Pi_{n}=M(L+cL^{-1})^{n}(c-L^{2}).
\end{eqnarray}
Its additional flows form a subalgebra of the Virasoro algebra, which generated by $\left\{L_{i} \mid i \in \mathbb{Z}_{\geq-1}\right\}$, satisfying
\begin{eqnarray}
\left[L_{i}, L_{j}\right]=(i-j)L_{i+j},\nonumber
\end{eqnarray}
where
\begin{eqnarray}
\frac{\partial }{\partial t_{i}^{*}} \mapsto L_{i+1}-4cL_{i-1}.\nonumber
\end{eqnarray}

{\rm \textbf{(\Rmnum{2}).}} For the $c$-$1$ constrained BKP hierarchy, its additional symmetry is given by \begin{eqnarray}\label{cBKPmodLax}
\frac{\partial \tilde{L}}{\partial t_{2 n+1}^*}=& \left[-(\tilde{\Pi}_{n})_{-}-\tilde{X}_{2 n+3}+4 c\tilde{ X}_{2 n+1}, \tilde{L}\right],
\end{eqnarray}
where
\begin{subequations}
\begin{alignat}{2}
\tilde{X}_1&=0, \nonumber\\
\tilde{X}_{2 n+1}&=& \sum_{i=1}^{\tilde{N}} \sum_{j=0}^{2 n-1}(2 j-(2 n-1))\left(\left(\tilde{L}+c \tilde{L}^{-1}\right)^{2 n-1-j}\left(\tilde{r}_i\right) \partial^{-1}\left(\left(\tilde{L}+c \tilde{L}^{-1}\right)^*\right)^j\left(\tilde{q}_{i, x}\right)\right.\nonumber\\
&\;&\left.-\left(\tilde{L}+c \tilde{L}^{-1}\right)^{2 n-1-j}\left(\tilde{q}_i\right) \partial^{-1}\left(\left(\tilde{L}+c \tilde{L}^{-1}\right)^*\right)^j\left(\tilde{r}_{i, x}\right)\right), \quad n=1,2,3, \cdots,\nonumber
\end{alignat}
\end{subequations}
and
\begin{eqnarray}
\tilde{\Pi}_{n}=\left(\tilde{M}\left(\tilde{L}+c \tilde{L}^{-1}\right)^{2 n+1}-(-1)^{2 n+1}\left(\tilde{L}+c \tilde{L}^{-1}\right)^{2 n} \tilde{M}\left(\tilde{L}+c \tilde{L}^{-1}\right)\right)\left(c-\tilde{L}^2\right).\nonumber
\end{eqnarray}
Its additional flows form a subalgebra of the Virasoro algebra, which generated by $\left\{\tilde{L}_{2 i} \mid i \in \mathbb{N}_{+}\right\}$, satisfying
\begin{eqnarray}
\left[\tilde{L}_{2 i}, \tilde{L}_{2 j}\right]=2(i-j) \tilde{L}_{2 i+2 j},\nonumber
\end{eqnarray}
where
\begin{eqnarray}
\frac{\partial}{\partial t_{2 i+1}^*} \mapsto 2\tilde{L}_{2 i+2}-8 c \tilde{L}_{2 i}.\nonumber
\end{eqnarray}
\end{thm}

\subsection{ The case of the $c$-$k$ constrained KP hierarchy}

We here prove Theorem \ref{thmadditonalsym} for the case {\rm \textbf{(\Rmnum{1})}} and get some results on additional flows. We first give two lemmas.

\begin{lemma}\label{ccKPLaxX}
The action of flows $\frac{\partial}{\partial t_{n}}$ of the $c$-$1$ constrained KP hierarchy on the $X_{m}$ are
\begin{eqnarray}\label{YLax}
\frac{\partial X_{m}}{\partial t_{n}}=\left[(L^{n})_{+}, X_{m} \right]_{-}.
\end{eqnarray}
\end{lemma}

\begin{proof}
For $m=0, 1, 2,$ \eqref{YLax} is clear. For $m=3, 4, 5,\cdots$,
\begin{eqnarray}
\frac{\partial X_{m}}{\partial t_{n}}&=&\sum_{i=1}^{N}\sum_{j=0}^{m-2}\left(j-\frac{1}{2}(m-2)\right)\frac{\partial }{\partial t_{n}}\left((L+cL^{-1})^{m-2-j}(q_{i})\right)\partial^{-1}\circ (L+cL^{-1})^{*j}(r_{i})\nonumber\\
&\;&+\sum_{i=1}^{N}\sum_{j=0}^{m-2}\left(j-\frac{1}{2}(m-2)\right)(L+cL^{-1})^{m-2-j}(q_{i})\partial^{-1}\circ \frac{\partial }{\partial t_{n}}\left((L+cL^{-1})^{*j}(r_{i})\right)\nonumber\\
&=&\sum_{i=1}^{N}\sum_{j=0}^{m-2}\left(j-\frac{1}{2}(m-2)\right)(L^{n})_{+}\left((L+cL^{-1})^{m-2-j}(q_{i})\right)\partial^{-1}\circ (L+cL^{-1})^{*j}(r_{i})\nonumber\\
&\;&-\sum_{i=1}^{N}\sum_{j=0}^{m-2}\left(j-\frac{1}{2}(m-2)\right)(L+cL^{-1})^{m-2-j}(q_{i})\partial^{-1}\circ (L^{n})^{*}_{+}\left((L+cL^{-1})^{*j}(r_{i})\right)\nonumber.
\end{eqnarray}
Then the proposition is proved by using the equation as
\begin{eqnarray}\label{Araty}
[K, f \partial^{-1} g]_{-}=K(f) \partial^{-1}g-f\partial^{-1}K^{*}(g),
\end{eqnarray}
where $K$ is an arbitrary differential operator and $f, g$ are arbitrary functions.
\end{proof}

Similarly to the result in \cite{EnriquezDa}, we can give the following lemma without proof.

\begin{lemma}\label{lemXL}
The Lax operator $L$ of the $c$-$1$ constrained KP hierarchy given by \eqref{ckKPLax} satisfies the following relation
\begin{eqnarray}
(L+c L^{-1})^{n}_{-}=\sum_{i=1}^{N}\sum_{j=0}^{n-1}(L+cL^{-1})^{n-j-1}(q_{i})\partial^{-1} \circ (L+cL^{-1})^{* j}(r_{i}),\;\;n=0, 1, 2,\cdots
\end{eqnarray}
and
\begin{eqnarray}
[Y, L+cL^{-1}]_{-}&=&-(L+cL^{-1})(M)\partial^{-1} \circ N+M \partial^{-1} \circ (L+cL^{-1})^{*}(N)\nonumber\\
&\;&+\sum_{i=1}^{N}\left(Y(q_{i})\partial^{-1} \circ r_{i}-q_{i}\partial^{-1}Y^{*}(r_{i})\right),
\end{eqnarray}
where
\begin{eqnarray}
Y=M\partial^{-1} \circ N.\nonumber
\end{eqnarray}
\end{lemma}

Then we can obtain the additional flows acting on $q_{i}(t)$ and $r_{i}(t)$ of the $c$-$1$ constrained KP hierarchy.

\begin{proposition}\label{thmact}
The additional flows acting on eigenfunction $q_{i}(t)$ and adjoint eigenfunction $r_{i}(t)$ of the $c$-$1$ constrained KP hierarchy are given by
\begin{eqnarray}
\frac{\partial q_{i}}{\partial t_{n}^{*}}&=&\left(\Pi_{n}\right)_{+}(q_{i})-X_{n+2}(q_{i})-\frac{n+2}{2}(L+cL^{-1})^{n+1}(q_{i})\nonumber\\
&\;&+4cX_{n}(q_{i})+2cn(L+cL^{-1})^{n-1}(q_{i}),\\
\frac{\partial r_{i}}{\partial t_{n}^{*}}&=&-\left(\Pi_{n}\right)^{*}_{+}(r_{i})+X^{*}_{n+2}(r_{i})-\frac{n+2}{2}\left((L+cL^{-1})^{*}\right)^{ n+1}(r_{i})\nonumber\\
&\;&-4cX_{n}(r_{i})+2cn\left((L+cL^{-1})^{*}\right)^{n-1}(r_{i}).
\end{eqnarray}
\end{proposition}

\begin{proof}
The equation \eqref{addLax} implies
\begin{eqnarray}
\frac{\partial L^{-1}}{\partial t_{n}^{*}}=[-\left(\Pi_{n})\right)_{-}-X_{n+2}+4cX_{n}, L^{-1} ],\nonumber
\end{eqnarray}
namely,
\begin{eqnarray}\label{addLaxL2}
\frac{\partial (L+cL^{-1})_{-}}{\partial t_{n}^{*}}=[-\left(\Pi_{n}\right)_{-}-X_{n+2}+4cX_{n}, (L+cL^{-1})]_{-}.
\end{eqnarray}
By Lemma \ref{lemXL}, we have
\begin{eqnarray}\label{XnL}
[X_{n}, L+cL^{-1}]_{-}&=&-(L+cL^{-1})^{n}_{-}+\sum_{i=1}^{N}\left(-q_{i}\partial ^{-1}\circ X_{n}^{*}(r_{i})+X_{n}(q_{i})\partial^{-1}\circ r_{i} \right)\nonumber\\
&\;&+\frac{n}{2}\sum_{i=1}^{N}\left((L+cL^{-1})^{n-1}(q_{i})\partial ^{-1} \circ r_{i}+q_{i} \partial^{-1} \circ \left((L+cL^{-1})^{*}\right)^{n-1}\right).
\end{eqnarray}
Using \eqref{XnL}, we can obtain that the right-hand side of \eqref{addLaxL2} is
\begin{eqnarray}\label{addLright}
\!\!\!\!\!\!\!&\;&4c\left((L+cL^{-1})^{n}\right)_{-}\!\!-\!\!\left((L+cL^{-1})^{n+2}\right)_{-}\!\!+\!\!\left[\left(\Pi_{n}\right)_{+},L+cL^{-1}\right]_{-}\!\!-\!\![X_{n+2}-4cX_{n},L+cL^{-1}]_{-}\!\!\!\!\nonumber\\
&=&\sum_{i=1}^{N}\left( \left(\Pi_{n}\right)_{+}(q_{i})\partial^{-1}\circ r_{i}-q_{i}\partial^{-1}\circ \left(\Pi_{n}\right)^{*}_{+}(r_{i})\right)\nonumber\\
&\;&+\sum_{i=1}^{N}\left(q_{i}\partial^{-1}X_{n+2}^{*}(r_{i})-X_{n+2}(q_{i})\partial^{-1}\circ r_{i}-4cq_{i}\partial^{-1}X_{n}^{*}(r_{i})-4cX_{n}(q_{i})\partial^{-1}\circ r_{i}\right)\nonumber\\
&\;&-\frac{n+2}{2}\sum_{i=1}^{N}\left((L+cL^{-1})^{n+1}(q_{i})\partial ^{-1} \circ r_{i}+q_{i} \partial^{-1} \circ \left((L+cL^{-1})^{*}\right)^{n+1}\right)\nonumber\\
&\;&+2cn\sum_{i=1}^{N}\left((L+cL^{-1})^{n-1}(q_{i})\partial ^{-1} \circ r_{i}+q_{i} \partial^{-1} \circ \left((L+cL^{-1})^{*}\right)^{n-1}\right),
\end{eqnarray}
and the left-hand side of \eqref{addLaxL2} is
\begin{eqnarray}\label{addLleft}
\sum_{i=1}^{N}\left(\frac{\partial q_{i}}{\partial t_{n}^{*}} \partial ^{-1} \circ r_{i}+q_{i}\partial ^{-1} \circ \frac{\partial r_{i}}{\partial t_{n}^{*}}\right).
\end{eqnarray}
We can prove the proposition via comparing \eqref{addLleft} and the right-hand sides of \eqref{addLright}.
\end{proof}

\begin{proposition}\label{corLi}
\begin{eqnarray}
\frac{\partial}{\partial t_{n}^{*}}(L+cL^{-1})^{j}(q_{i})\!\!&=&\!\!\left(\Pi_{n}\right)_{+}(L+cL^{-1})^{j}(q_{i})-\left(j+\frac{n+2}{2}\right)(L+cL^{-1})^{n+j+1}(q_{i})\nonumber\\
&\;&-X_{n+2}(L+cL^{-1})^{j}(q_{i})+4cX_{n}(L+cL^{-1})^{j}(q_{i})\nonumber\\
&\;&+2c(2j+n)(L+cL^{-1})^{n+j-1}(q_{i}),\nonumber\\
\frac{\partial}{\partial t_{n}^{*}}(L+cL^{-1})^{* j}(r_{i})\!\!&=&\!\!-\left(\Pi_{n}\right)^{*}_{+}(L+cL^{-1})^{* j}(r_{i})-\left(j+\frac{n+2}{2}\right)\left((L+cL^{-1})^{*}\right)^{ n+j+1}(r_{i})\nonumber\\
&\;&+X^{*}_{n+2}(L+cL^{-1})^{* j}(r_{i})-4cX_{n}^{*}(L+cL^{-1})^{* j}(r_{i})\nonumber\\
&\;&+2c(2j+n)\left((L+cL^{-1})^{*}\right)^{n+j-1}(r_{i}).\nonumber
\end{eqnarray}
\end{proposition}
\begin{proof}
The proposition is proved by using equations
\begin{eqnarray}
\frac{\partial}{\partial t_{n}^{*}}(L+cL^{-1})^{j}(q_{i})&=&\frac{\partial (L+cL^{-1})^{j}}{\partial t_{n}^{*}}(q_{i})+(L+cL^{-1})^{j}\left(\frac{\partial}{\partial t_{n}^{*}}(q_{i})\right),\nonumber\\
\frac{\partial}{\partial t_{n}^{*}}(L+cL^{-1})^{* j}(r_{i})&=&\frac{\partial (L+cL^{-1})^{* j}}{\partial t_{n}^{*}}(r_{i})+(L+cL^{-1})^{* j}\left(\frac{\partial}{\partial t_{n}^{*}}(r_{i})\right)\nonumber.
\end{eqnarray}
\end{proof}

Now we give the following lemma in order to prove the theorem.

\begin{lemma}\label{lemXkn}
Let $\Xi_{n}=-(\Pi_{n})_{<0}-X_{n+2}+4cX_{n}$, then $\Xi_{n} (n=0, 1, 2, \cdots)$ satisfy the following relations
\begin{eqnarray}
\frac{\partial \Xi_{n}}{\partial t_{m}^{*}}-\frac{\partial \Xi_{m}}{\partial t_{n}^{*}}+\left[\Xi_{n}, \Xi_{m} \right]=-(n-m)\Xi_{m+n+1}+4c(n-m)\Xi_{n+m-1}.\nonumber
\end{eqnarray}
\end{lemma}

\begin{proof}
By Proposition \ref{corLi} and \eqref{Araty}, we have
\begin{eqnarray}
&\;&\frac{\partial \Xi_{n}}{\partial t_{m}^{*}}-\frac{\partial \Xi_{m}}{\partial t_{n}^{*}}+\left[\Xi_{n}, \Xi_{m} \right] \nonumber\\
&=&\left[\Pi_{m}, \Pi_{n}\right]_{<0}-(n-m)X_{m+n+3}-8c(n-m)X_{n+m+1}+16c^{2}(n-m)X_{n+m-1} \nonumber\\
&=&-(n-m)\Xi_{m+n+1}+4c(n-m)\Xi_{n+m-1}\nonumber.
\end{eqnarray}
In the above calculations we have used the following two main techniques, i.e.,
\begin{eqnarray}
&\;&\sum_{i=1}^{N}\sum_{j=0}^{m}\left(j-\frac{m}{2}\right)X_{n+2}\left((L+cL^{-1})^{m-j}(q_{i})\right)\partial^{-1} \circ ((L+cL^{-1})^{*})^{j}(r_{i})\nonumber\\
&\;&+\sum_{i=1}^{N}\sum_{j=0}^{m}\left(j-\frac{m}{2}\right)\left((L+cL^{-1})^{m-j}(q_{i})\right)\partial^{-1} \circ X_{n+2}^{*}((L+cL^{-1})^{*})^{j}(r_{i})\nonumber\\
&\;&+\sum_{i=1}^{N}\sum_{j=0}^{n}\left(j-\frac{n}{2}\right)X_{m+2}\left((L+cL^{-1})^{n-j}(q_{i})\right)\partial^{-1} \circ ((L+cL^{-1})^{*})^{j}(r_{i})\nonumber\\
&\;&+\sum_{i=1}^{N}\sum_{j=0}^{n}\left(j-\frac{n}{2}\right)\left((L+cL^{-1})^{n-j}(q_{i})\right)\partial^{-1} \circ X_{m+2}^{*}((L+cL^{-1})^{*})^{j}(r_{i})\nonumber\\
&=&[X_{n+2}, X_{m+2}],\nonumber
\end{eqnarray}
and
\begin{eqnarray}
&\;&\sum_{i=1}^{N}\sum_{j=0}^{m+n+1}\left(j-\frac{n}{2}\right)\left(n-j+\frac{m}{2}+1\right)(L+cL^{-1})^{n-2-j}(q_{i})\partial^{-1}\circ (L+cL^{-1})^{*j}(r_{i})\nonumber\\
&\;&+\sum_{i=1}^{N}\sum_{j=0}^{m+n+1}\left(j-\frac{n}{2}-m-1\right)\left(n-\frac{m}{2}\right)(L+cL^{-1})^{n-2-j}(q_{i})\partial^{-1}\circ (L+cL^{-1})^{*j}(r_{i})\nonumber\\
&=&(n-m)X_{m+n+3}.\nonumber
\end{eqnarray}
\end{proof}

\noindent
\textit{Proof of Theorem \ref{thmadditonalsym} for the case {\rm \textbf{(\Rmnum{1})}}.}

On the one hand, we need proof that the additional flows \eqref{addLax} commute with the flows of the $c$-$1$ constrained KP hierarchy.

Using Lemma \ref{ccKPLaxX}, \eqref{ckKPeq} and \eqref{addLax}, we have
\begin{eqnarray}
&\;&\left[\frac{\partial}{\partial t_{m}}, \frac{\partial }{\partial t_{n}^{*}} \right]L\nonumber\\
&=&\left[\left[(L^{m})_{+},-\Pi_{n}-X_{n+2}+4cX_{n} \right]_{-}, L \right]+\left[\left(-\Pi_{n}-X_{n+2}+4cX_{n}\right)_{-},\left[(L^{m})_{+}, L\right] \right]\nonumber\\
&\;&-\left[\left[\left(-\Pi_{n}-X_{n+2}+4cX_{n}\right)_{-},L^{m}\right]_{+}, L \right]-\left[(L^{m})_{+},\left[\left(-\Pi_{n}-X_{n+2}+4cX_{n}\right)_{-}, L \right] \right]\nonumber\\
&=&\left[\left[(L^{m})_{+},\left(-\Pi_{n}-X_{n+2}+4cX_{n}\right)_{-}\right], L \right]+\left[\left(-\Pi_{n}-X_{n+2}+4cX_{n}\right)_{-},\left[(L^{m})_{+}, L\right] \right]\nonumber\\
&\;&-\left[(L^{m})_{+},\left[\left(-\Pi_{n}-X_{n+2}+4cX_{n}\right)_{-}, L \right] \right]\nonumber.
\end{eqnarray}
The last equation vanishes because of the Jacobi identity.

On the other hand, using Lemma \ref{lemXkn}, for $m, n=0, 1, 2, \cdots$, we can get
\begin{eqnarray}
\left[ \frac{\partial }{\partial t_{m}^{*}}, \frac{\partial }{\partial t_{n}^{*}}\right]=-(n-m)\frac{\partial}{\partial t_{n+m+1}^{*}}+4c(n-m)\frac{\partial}{\partial t_{n+m-1}^{*}}.\nonumber
\end{eqnarray}
For $m=0, 1, 2, \cdots$, we get a map $\varphi$ defined by
\begin{eqnarray}
\frac{\partial }{\partial t_{m}^{*}} \mapsto L_{m+1}-4cL_{m-1},\nonumber
\end{eqnarray}
where $\left\{ L_{i}|\;i\in \mathbb{Z}_{\geqslant -1} \right \}$ satisfy the algebraic relations
\begin{eqnarray}
\left[L_{i}, L_{j}\right]=(i-j)L_{i+j}.\nonumber
\end{eqnarray}
Since
\begin{eqnarray}
\varphi \left(\left[\frac{\partial }{\partial t_{m}^{*}}, \frac{\partial }{\partial t_{n}^{*}}\right]\right)&=&-(n-m)\varphi \left(\frac{\partial}{\partial t_{n+m+1}^{*}}\right)+4c(n-m)\varphi\left(\frac{\partial}{\partial t_{n+m-1}^{*}}\right)\nonumber\\
&=&(m-n)\left(L_{n+m+2}-8cL_{n+m}+16c^{2}L_{n+m-2}\right)\nonumber\\
&=&\left[\varphi\left(\frac{\partial }{\partial t_{m}^{*}} \right), \varphi \left(\frac{\partial }{\partial t_{n}^{*}} \right) \right].\nonumber
\end{eqnarray}
the map is an homomorphism of the Lie algebra.\hfill\qedsymbol

\subsection{ The case of the $c$-$k$ constrained BKP hierarchy}

We here prove Theorem \ref{thmadditonalsym} for the case {\rm \textbf{(\Rmnum{2})}} and get some results on additional flows. We first give two lemmas.

\begin{lemma}\label{propcBKPL}
The operator $\tilde{L}$ of the $c$-$1$ constrained BKP hierarchy satisfies
\begin{eqnarray}
\left(\tilde{L}+c \tilde{L}^{-1}\right)_{-}^n&=& \sum_{i=1}^{\tilde{N}} \sum_{j=0}^{n-1}\left(\left(\tilde{L}+c \tilde{L}^{-1}\right)^{n-j-1}\left(\tilde{r}_i\right) \partial^{-1}\left(\left(\tilde{L}+c \tilde{L}^{-1}\right)^*\right)^j\left(\tilde{q}_{i, x}\right)\right.\nonumber\\
&\;&\left.-\left(\tilde{L}+c \tilde{L}^{-1}\right)^{n-j-1}\left(\tilde{q}_i\right) \partial^{-1}\left(\left(\tilde{L}+c \tilde{L}^{-1}\right)^*\right)^j\left(\tilde{r}_{i, x}\right)\right)\nonumber.
\end{eqnarray}
\end{lemma}

\begin{lemma}\label{propcBKPX}
The flows $\frac{\partial}{\partial t_l}$ of the $c$-$1$ constrained BKP hierarchy acting on $\tilde{X}_{2n+1}$ can be written as
\begin{eqnarray}
\frac{\partial \tilde{X}_{2 n+1}}{\partial t_l}=\left[(\tilde{L}^l)_{+}, \tilde{X}_{2 n+1}\right]_{-},\nonumber
\end{eqnarray}
where $l=1,3,5, \cdots$.
\end{lemma}

\begin{proof}
\begin{eqnarray}
\frac{\partial \tilde{X}_{2 n+1}}{\partial t_l}&=& \sum_{i=1}^{\tilde{N}} \sum_{j=0}^{2 n-1}(2 j-(2 n-1)) \frac{\partial}{\partial t_l}\left(\tilde{L}+c \tilde{L}^{-1}\right)^{2 n-1-j}\left(\tilde{r}_i\right) \partial^{-1}\left(\left(\tilde{L}+c \tilde{L}^{-1}\right)^*\right)^j\left(\tilde{q}_{i, x}\right) \nonumber\\
&\;&+\sum_{i=1}^{\tilde{N}} \sum_{j=0}^{2 n-1}(2 j-(2 n-1))\left(\tilde{L}+c \tilde{L}^{-1}\right)^{2 n-1-j}\left(\tilde{r}_i\right) \partial^{-1} \frac{\partial}{\partial t_l}\left(\left(\tilde{L}+c \tilde{L}^{-1}\right)^*\right)^j\left(\tilde{q}_{i, x}\right) \nonumber\\
&\;&\left.-\sum_{i=1}^{\tilde{N}} \sum_{j=0}^{2 n-1}(2 j-(2 n-1)) \frac{\partial}{\partial t_l}\left(\tilde{L}+c \tilde{L}^{-1}\right)^{2 n-1-j}\left(\tilde{q}_i\right) \partial^{-1}\left(\left(\tilde{L}+c \tilde{L}^{-1}\right)^*\right)^j\left(\tilde{r}_{i, x}\right)\right) \nonumber\\
&\;&\left.-\sum_{i=1}^{\tilde{N}} \sum_{j=0}^{2 n-1}(2 j-(2 n-1))\left(\tilde{L}+c \tilde{L}^{-1}\right)^{2 n-1-j}\left(\tilde{q}_i\right) \partial^{-1} \frac{\partial}{\partial t_l}\left(\left(\tilde{L}+c \tilde{L}^{-1}\right)^*\right)^j\left(\tilde{r}_{i, x}\right)\right).\nonumber
\end{eqnarray}
Then this proposition can be proved by using \eqref{Araty}.
\end{proof}

Now, with reference to Lemma \ref{propcBKPL}, we can give the actions of the additional flows defined by \eqref{cBKPmodLax} on eigenfunctions.

\begin{proposition}\label{propBKPqr}
The actions of the defined additional flows $\frac{\partial}{\partial t_{2n+1}^*}$ on eigenfunctions $\tilde{q}_i$ and $\tilde{r}_i$ are
\begin{eqnarray}
\frac{\partial \tilde{r}_i}{\partial t_{2 n+1}^*}&=&-(\tilde{\Pi}_{n})_{+}\left(\tilde{r}_i\right)-(2 n+3)\left(\tilde{L}+c \tilde{L}^{-1}\right)^{2 n+2}\left(\tilde{r}_i\right)-\tilde{X}_{2 n+3}\left(r_i\right) \nonumber\\
&\;&+4 c(2 n+1)\left(\tilde{L}+c \tilde{L}^{-1}\right)^{2 n}\left(\tilde{r}_i\right)+4 c \tilde{X}_{2 n+1}\left(\tilde{r}_i\right) \nonumber\\
\frac{\partial \tilde{q}_i}{\partial t_{2 n+1}^*}&=&-(\tilde{\Pi}_{n})_{+}\left(\tilde{q}_i\right)-(2 n+3)\left(\tilde{L}+c \tilde{L}^{-1}\right)^{2 n+2}\left(\tilde{q}_i\right)-\tilde{X}_{2 n+3}\left(\tilde{q}_i\right) \nonumber\\
&\;&+4 c(2 n+1)\left(\tilde{L}+c \tilde{L}^{-1}\right)^{2 n}\left(\tilde{q}_i\right)+4 c \tilde{X}_{2 n+1}\left(\tilde{q}_i\right) \nonumber\\
\frac{\partial \tilde{r}_{i, x}}{\partial t_{2 n+1}^*}&=&-(\tilde{\Pi}_{n}^*)_{+}\left(\tilde{r}_{i, x}\right)-(2 n+3)\left(\left(\tilde{L}+c \tilde{L}^{-1}\right)^*\right)^{2 n+2}\left(\tilde{r}_{i, x}\right)+\tilde{X}_{2 n+3}^*\left(\tilde{r}_{i, x}\right) \nonumber\\
&\;&+4 c(2 n+1)\left(\left(\tilde{L}+c \tilde{L}^{-1}\right)^*\right)^{2 n}\left(\tilde{r}_i\right)-4 c \tilde{X}_{2 n+1}^*\left(\tilde{r}_{i, x}\right) \nonumber\\
\frac{\partial \tilde{q}_{i, x}}{\partial t_{2 n+1}^*}&=&-(\tilde{\Pi}_{n}^*)_{+}\left(\tilde{q}_{i, x}\right)-(2 n+3)\left(\left(\tilde{L}+c \tilde{L}^{-1}\right)^*\right)^{2 n+2}\left(\tilde{q}_{i, x}\right)+\tilde{X}_{2 n+3}^*\left(\tilde{q}_{i, x}\right) \nonumber\\
&\;&+4 c(2 n+1)\left(\left(\tilde{L}+c \tilde{L}^{-1}\right)^*\right)^{2 n}\left(\tilde{q}_i\right)-4 c \tilde{X}_{2 n+1}^*\left(\tilde{q}_{i, x}\right).\nonumber
\end{eqnarray}
where $n=0,1,2, \cdots$.
\end{proposition}

\begin{lemma}\label{propT}
Let $\tilde{\Xi}_{2 n+1}=-(\tilde{\Pi}_{n})_{-}-\tilde{X}_{2 n+3}+4 c \tilde{X}_{2 n+1}$, then
\begin{eqnarray}
\frac{\partial \tilde{\Xi}_{2 n+1}}{\partial t_{2 k+1}^*}-\frac{\partial \tilde{\Xi}_{2 k+1}}{\partial t_{2 n+1}^*}+\left[\tilde{\Xi}_{2 n+1}, \tilde{\Xi}_{2 k+1}\right]=-4(n-k) \tilde{\Xi}_{2 n+2 k+3}+16 c(n-k) \tilde{\Xi}_{2 n+2 k-3},\nonumber
\end{eqnarray}
where $n, k=0,1,2, \cdots$.
\end{lemma}

\begin{proof}
Referring to the proof of Lemma \ref{lemXkn}, especially those two processing techniques, we can prove this proposition by using Lemma \ref{propcBKPX}
\end{proof}

\noindent
\textit{Proof of Theorem \ref{thmadditonalsym} for the case {\rm \textbf{(\Rmnum{2})}}.}

For one thing, by Proposition \ref{propBKPqr}, we have
\begin{eqnarray}
\left[\frac{\partial}{\partial t_{2 m+1}}, \frac{\partial}{\partial t_{2 n+1}^*}\right]\tilde{L}&=& {\left[\left[\left(\tilde{L}^{2 m+1}\right)_{+},\left(-\tilde{\Pi}_{n}-\tilde{X}_{2 n+3}+4 c \tilde{X}_{2 n+1}\right)_{-}\right]_{-}, \tilde{L}\right] } \nonumber\\
&\;&+\left[\left(-\Tilde{\Pi}_{n}-X_{2 n+3}+4 c \tilde{X}_{2 n+1}\right)_{-},\left[\left(\tilde{L}^{2 m+1}\right)_{+}, \tilde{L}\right]\right] \nonumber\\
&\;&-\left[\left(-\Tilde{\Pi}_{n}-\tilde{X}_{2 n+3}+4 c \tilde{X}_{2 n+1}\right)_{-},\left[\left(\tilde{L}^{2 m+1}\right)_{+}, \tilde{L}\right]\right]\nonumber\\
&\;&-\left[\left(\tilde{L}^{2 m+1}\right)_{+},\left[\left(-\tilde{\Pi}_{n}-\tilde{X}_{2 n+3}+4 c \tilde{X}_{2 n+1}\right), \tilde{L}\right]\right] \nonumber\\
&=& {\left[\left(-\Tilde{\Pi}_{n}-\tilde{X}_{2 n+3}+4 c \tilde{X}_{2 n+1}\right)_{-},\left[\left(\tilde{L}^{2 m+1}\right)_{+}, \tilde{L}\right]\right] } \nonumber\\
&\;&-\left[\left[\left(-\Tilde{\Pi}_{n}-\tilde{X}_{2 n+3}+4 c \tilde{X}_{2 n+1}\right)_{-},\left(\tilde{L}^{2 m+1}\right)_{+}\right], \tilde{L}\right] \nonumber\\
&\;&-\left[\left(\tilde{L}^{2 m+1}\right)_{+},\left[\left(-\tilde{\Pi}_{n}-\tilde{X}_{2 n+3}+4 c \tilde{X}_{2 n+1}\right)_{-}, \tilde{L}\right]\right]\nonumber\\
&=&0.\nonumber
\end{eqnarray}
We can obtain that the flows $\frac{\partial}{\partial t_{2 m+1}}$ and the additional flows $\frac{\partial}{\partial t_{2 n+1}^*}$ of the $c$-$1$ constrained BKP hierarchy satisfy
\begin{eqnarray}
\left[\frac{\partial}{\partial t_{2 n+1}^*}, \frac{\partial}{\partial t_{2 m+1}}\right]=0,\nonumber
\end{eqnarray}
where $n, m=0,1,2, \cdots$.

For another, Lemma \ref{propT} means that for $n, k=0,1,2, \cdots$, we have
\begin{eqnarray}
\left[\frac{\partial}{\partial t_{2 k+1}^*}, \frac{\partial}{\partial t_{2 n+1}^*}\right] L&=& {\left[\frac{\partial \tilde{\Xi}_{2 n+1}}{\partial t_{2 k+1}^*}-\frac{\partial \tilde{\Xi}_{2 k+1}}{\partial t_{2 n+1}^*}+\left[\tilde{\Xi}_{2 n+1}, \tilde{\Xi}_{2 k+1}\right], \tilde{L}\right] } \nonumber\\
&=&-2(2 n-2 k) \frac{\partial \tilde{L}}{\partial t_{2 n+2 k+3}^*}+8 c(2 n-2 k) \frac{\partial \tilde{L}}{\partial t_{2 n+2 k+1}^*}.\nonumber
\end{eqnarray}
For $k=0,1,2, \cdots$, we give a map $\tilde{\varphi}$ defined as
\begin{eqnarray}
\tilde{\varphi}: \frac{\partial}{\partial t_{2 k+1}^*} \mapsto 2\tilde{L}_{2 k+2}-8 c \tilde{L}_{2 k}, \nonumber
\end{eqnarray}
where $\tilde{L}_{2 k}$ satisfy the relation
\begin{eqnarray}
\left[\tilde{L}_{2 i}, \tilde{L}_{2 j}\right]=2(i-j) \tilde{L}_{2 i+2 j}.\nonumber
\end{eqnarray}
The map is an homomorphism since
\begin{eqnarray}
\tilde{\varphi}\left(\left[\frac{\partial}{\partial t_{2 k+1}^*}, \frac{\partial}{\partial t_{2 n+1}^*}\right]\right) &=&\tilde{\varphi}\left(-2(2 n-2 k) \frac{\partial}{\partial t_{2 n+2 k+3}^*}+8 c(2 n-2 k) \frac{\partial}{\partial t_{2 n+2 k+1}^*}\right) \nonumber\\
&=&\left[\tilde{\varphi}\left(\frac{\partial}{\partial t_{2 k+1}^*}\right), \tilde{\varphi}\left(\frac{\partial}{\partial t_{2 n+1}^*}\right)\right].\nonumber
\end{eqnarray}
\hfill\qedsymbol

\section{Conclusions}

In this paper, we have studied the $c$-$k$ constrained KP and BKP hierarchies. The $c$-$k$ constrained BKP hierarchy defined in \eqref{ckBKPeq}. In Theorem \ref{thmFermionic}, we will give the Fermionic picture of the $c$-$k$ constrained KP and BKP hierarchies by using the Clifford algebras $\mathcal{A}$ and $\mathcal{A}_B$.  Some solutions of the $c$-$k$ constrained KP and BKP hierarchies are given by using the Fermion operators in Theorem \ref{thmBEsolution}. The additional symmetry of the $c$-$1$ constrained KP and BKP hierarchies are studied in Theorem \ref{thmadditonalsym}. For $k>1$ we can discuss it in a similar way.

We believe that the results here are important. Just as we know, it is not applicable for the whole integrable hierarchies since it is usually very difficult to obtain explicit differential equations. Therefore various reductions will become very crucial. Another important point is that the Fermonic pictures of the $c$-$k$ constrained KP and BKP hierarchies are very fundamental, when considering some essential properties of the integrable hierarchies. The last point is that we give the new hidden realizations of  the Virasoro  algebras  in the integrable systems.
\section*{\bf Acknowledgements}
This work is supported by the National Natural Science Foundation of China under Grant Nos. 12171133, 12271136, 12171132 and 12171472, and the Anhui Province Natural Science Foundation No. 2008085MA05.

The authors would like to give our heartfelt thanks to Mr.Minghao Wang who have kindly provided our assistance and companionship in the course of preparing this paper. We sincerely hope that he will soon be rallied from his coma.

\end{CJK*}
\end{document}